%% file: paper.tex
\documentclass[a4paper,UKenglish]{lipics-v2016}
 
\usepackage{microtype}
\usepackage{xspace}

\input{macro} 

\title{Strongly Normalizing Audited Computation}

\author[1]{Wilmer Ricciotti}
\author[1]{James Cheney}
\affil[1]{Laboratory for Foundations of Computer Science, University
  of Edinburgh\\
  \texttt{[wricciot|jcheney]@inf.ed.ac.uk}}
\authorrunning{W. Ricciotti and J. Cheney} 

\Copyright{Wilmer Ricciotti and James Cheney}


\begin{document}

\maketitle

\begin{abstract}
  Auditing is an increasingly important operation for computer
  programming, for example in security (e.g. to enable history-based
  access control) and to enable reproducibility and accountability
  (e.g.  provenance in scientific programming).  Most proposed
  auditing techniques are ad hoc or treat auditing as a second-class,
  extralinguistic operation; logical or semantic foundations for
  auditing are not yet well-established. \emph{Justification Logic}
  (JL) offers one such foundation; Bavera and Bonelli introduced a
  computational interpretation of JL called \JLh that supports
  auditing. However, \JLh is technically complex and strong
  normalization was only established for special cases.  In addition,
  we show that the equational theory of \JLh is inconsistent.
 We introduce a new calculus \JLr that is simpler than \JLh, consistent, and
 strongly normalizing.  Our proof of strong normalization is
 formalized in Nominal Isabelle.
\end{abstract}

\section{Introduction}\label{sec:intro}
Auditing is of increasing importance in many areas, including in
security (e.g. history-based access control~\cite{Abadi2003},
evidence-based audit~\cite{vaughan08csf,jia08icfp}) and to provide
reproducibility or accountability for scientific computations
(e.g. provenance tracking~\cite{Acar2013}).  To date, most approaches
to auditing propose to instrument programs in some \emph{ad hoc} way,
to provide a record of the computation, sometimes called
\emph{provenance}, \emph{audit log}, \emph{trace}, or \emph{trail},
which is typically meant to be inspected after the computation
finishes.  Foundations of provenance or auditing have not been fully
developed, and auditing typically remains a second-class citizen (or
even an alien concept), in that trails are not accessible to the
program itself during execution, but only to an external monitoring
layer.  We seek logical and semantic foundations for
\emph{self-auditing programs} that provide first-class recording and
auditing primitives.

In this paper, we explore how ideas from \emph{justification
  logic}~\cite{Artemov2008} (a generalization of the \emph{Logic of
  Proofs}~\cite{Artemov2001}) can provide such a foundation.
\emph{Justification Logic} is the general name for a family of logics
that refine the epistemic reading of modal necessity $\Box A$ (``$A$
is known'') by
indexing such formulas with \emph{justifications}, or (descriptions
of) evidence that $A$ holds.  So, the formula $\AU{a}{A}$ can be read
as ``justification $a$ is evidence (proof) establishing $A$''.  In
elementary forms of justification logic, the justifications consist of
uninterpreted constants that can be combined with algebraic operations
(e.g. formal sum and product). A representative axiom of justification
logic is the evidence-aware version of Modus Ponens:
\[\AU{a}{(A\impp B)} \impp \AU{b}{A} \impp \AU{a\cdot b}{B}
\]
which can be read as ``if $a$ proves $A\impp B$ and $b$ proves $A$
then $a$ together with $b$ proves $B$''. Another representative axiom
is the reflection principle:
\[\AU{a}{A} \impp \AU{!a}{\AU{a}{A}}\]
which can be read as ``if $a$ proves $A$ then $!a$ proves that $a$
proves $A$'', where $!$ is a \emph{reflection} operator on
justifications.  Notice that in both cases, if we ignore the evidence
annotations $a,b$, we get standard axioms ``K'' and ``4'' of modal
logic respectively~\cite{Lewis1959}.

This paper explores using (constructive) justification logic as a
foundation for programming with auditing, building on of previous work
on calculi for modal logic~\cite{Pfenning2001} and justification
logic~\cite{Artemov2007,Bavera2015}.  Artemov and
Bonelli~\cite{Artemov2007} introduced the \emph{intensional lambda
  calculus} $\lambda^\mathbf{I}$, which builds upon Pfenning and
Davies' type theory for modal logic~\cite{Pfenning2001} by recording
\emph{proof codes} in judgements and modalities; proof codes are a
special case of proof terms.  In $\lambda^\mathbf{I}$, ``trails'' that
witness the equivalence of proof codes were introduced, but these
trails could not be inspected by other constructs.  More recently, Bavera and Bonelli~\cite{Bavera2015} introduced \JLh,
an extension of $\lambda^\mathbf{I}$ that allowed \emph{trail
  inspection}, which can audit a computation's trail of evaluation
from the currently-enclosing box constructor up to the point of
inspection.  

Trail inspection is a first-class construct of \JLh, and can be
performed at arbitrary program points (including inside another trail
inspection).  Bavera and Bonelli outlined how this construct can be
used to implement \emph{history-based access control}~\cite{Abadi2003}
policies in which access control decisions are made based on the
complete previous call history.  Trail inspection could also be used
for profiling, for example by counting the number of calls of each
function in a subcomputation, or for more general
auditing~\cite{amir-mohamedian2016}, for example by comparing the
trails of multiple computations to look for differences or anomalies
in their execution history.  Furthermore, all of these applications
can be supported in a single language, rather than requiring ad hoc
changes to accommodate different monitoring semantics.  (It is worth
noting that these applications typically use data structures such as
sets or maps that are not provided in \JLh; however, they are straightforward
to add.)

Bavera and Bonelli proved basic metatheoretic properties of
\JLh including subject reduction, and gave a partial proof of strong
normalization.  The latter proof considered only terms in which there
was no nesting of box constructors.  Although they conjectured that
strong normalization holds for the full system, the general case was
left as an open question. 
However, \JLh is a somewhat complex system, due to subtle aspects of
the metatheory of trail inspection.
Trails are viewed as proof terms witnessing equivalence of proof
codes.  Proof codes include trail inspections, whose results depend on
the history of the term being executed.  
\JLh's trails are viewed as equality proofs and are thus
subject to symmetry: operationally, the language provides a trail
constructor $\qsym(q)$ that converts a trail $q$ of type $s = t$ to
one of type $t = s$. Since the reduct of an outer trail inspection
depends on a non-local, volatile context (the trail declared in the
nearest outer box constructor) symmetry makes it possible to equate any two codes
of the same type: in other words, the theory of trails becomes
inconsistent.

To avoid this problem, Bavera and Bonelli annotated trail inspections
with affine trail variables $\alpha, \alpha'$ representing the
volatile context, and trail inspection is then denoted by
$\alpha\theta$, where $\theta$ is a collection of functions specifying
a structural recursion over the trail. 
Nevertheless, in \JLh we can prove any two codes $s$ and $t$
equivalent.  A detailed example, along with the rules of \JLh, are
presented in Appendix~\ref{app:jlh}; here, we sketch the derivation:
\begin{center}\small
  \begin{prooftree}
    \AxiomC{$s = q_1\theta$}
    \AxiomC{$q_1\theta = \alpha\theta$}
    \BinaryInfC{$s = \alpha\theta$}
    \AxiomC{$t = q_2\theta$}
    \AxiomC{$q_2\theta = \alpha\theta$}
    \BinaryInfC{$t = \alpha\theta$}
    \UnaryInfC{$\alpha\theta = t$}
    \RightLabel{(*)}
    \BinaryInfC{$s = t$}
  \end{prooftree}
\end{center}
where $q_1$, $q_2$ and $\theta$ are chosen so that $s = q_1\theta$ and
$t = q_2\theta$, which is possible whenever
$s$ and $t$ have the same type.
The problem
seems to originate from the transitivity rule marked with (*), which
constitutes an exception to the affineness of trail variables, since
the same proof code must necessarily appear as the right-hand side of
the first premise and as the left-hand side of the second one. This
allows trail inspection on $\alpha$ to be triggered twice.  It should
be noted that this property does \emph{not} contradict any proved
results by Bavera and Bonelli; instead, it can be viewed as an
(undesirable) form of proof-irrelevance: that is,
$\AU{s}A \to \AU{t}A$ is inhabited for any well-formed proof codes
$s,t$.  Of course, if this is the case then the stated justification
$s$ for $\AU{s}A$ need have nothing with to do its actual proof, defeating
the purpose of introducing the proof codes in the first place.



We present a new calculus, \JLr, that overcomes these issues with
\JLh, and provides a 
more satisfactory foundation for
audited computation, including the following contributions:
\begin{itemize}
\item We dispense with \JLh's trail variables, instead adopting a
  trail inspection construct $\tinsp(\theta)$ that can be used
  anywhere in the program. We also adjust \JLh's treatment of trails
  by no longer viewing them as witnesses to proof code \emph{equality}
  but rather to (directed) \emph{reducibility}.  At a technical level,
  this simply consists of dropping the symmetry axiom and trail
  form. Viewing trails as asymmetric reducibility proofs means that it
  is no longer inconsistent to relate a trail inspection redex with
  all of its possible contracta nondeterministically, and it also
  obviates the need for trail variables (or ordinary variables) to be
  affine; once we lift these restrictions, however, the rationale for
  trail variables evaporates. We present the improved system, and
  summarize its metatheory (including a proof that its equational
  theory is consistent) in Section~\ref{sec:calculus}.
\item To show that $\JLr$ is well-behaved, we prove strong
  normalization (SN). Our proof of SN actually shows a stronger
  result: SN holds even if we conservatively assume that trail
  inspection operations can choose \emph{any} well-formed trail; this
  is a useful insight because it means that we can extend the
  result to more realistic languages with more complex trails. We
  summarize this proof, which we fully formalized
  in Nominal Isabelle\footnote{The full formalization can be found at \url{http://www.wilmer-ricciotti.net/FORMAL/ccau.tgz}. A more detailed discussion can be found in the extended version of this paper.}
\end{itemize}

Section~\ref{sec:related} discusses related work: in particular, we
argue that our approach can be used to prove strong normalization for
\JLh as well, even in the presence of the symmetry rule.


\section{A history-aware calculus}
\label{sec:calculus}
\subsection{Overview}
In this section we describe \JLr, an extension of the simply-typed
lambda calculus allowing for several forms of auditing, thanks to an
internalized notion of \emph{computation history}.  It is based in
many respects on \JLh, but since our preferred notation differs in
important details we give a self-contained presentation.

Operationally, we add to the simply typed $\lambda$-calculus a new term form $\bang_q M$, where $M$ is
an expression and $q$ is a \emph{trail} containing a symbolic
description of the past history of $M$. History-aware expressions will
be given their own type: if $M$ has type $A$, and $q$ is an
appropriate trail showing how $M$ has reduced, then the type of
$\bang_q M$ will be indicated by $\AU{s} A$; the symbol `$\bang$' is
called ``bang''.
In our language, we will initially evaluate expressions in the form $\bang~M$, where the absence of a subscript indicates that no computation has happened within $M$ yet. As we reduce $M$, the bang will be annotated with a subscript, for example
\[
\bang~M \red \bang_q N
\]
where $q$ must explain how $N$ was obtained from the initial expression $M$.

As a slightly more involved example, we consider the following expression:
\[
\bang~\iftrm~(1 \equest 2)~M~N
\]
Here, $\equest$ represents an equality test between integers, while $M$ and $N$ represent respectively the `then' and `else' branches of the $\iftrm$ statement. To evaluate the expression, we first reduce the equality test to $\# F$ (boolean false), and the $\iftrm$ to its `else' branch. In \JLr, this computation will be described as:
\[
\bang~\iftrm~(1 \equest 2)~M~N \red \bang_{\cQ[q]} \iftrm~\# F~M~N \red
\bang_{\qtrans(\cQ[q],q')} N
\]
where $q$, $\cQ$, and $q'$ are defined in such a way that:
\begin{itemize}
\item $q$ explains how $1 \equest 2$ reduces to $\# F$
\item $\cQ$ is a context indicating that $q$ was applied to the guard of an $\iftrm$ statement -- in other words, it describes a congruence rule of reduction
\item $q'$ describes the reduction of $\iftrm~\# F \cdots$ to its `else' branch
\item $\qtrans$ combines its two arguments by means of transitivity.
\end{itemize}

To make use of history-aware expressions, we will also provide constructs to access their contents or to inspect their history. 

\subsection{Syntax}
We now give the formal definition of $\JLr$. The language includes
three syntactic classes, which we have already presented informally in
the overview:
\begin{itemize}
\item \emph{codes} $s,t$ correspond to (extended) lambda expressions which have not been evaluated yet: in other words, they only contain bangs with no history
\item \emph{trails} $q,q'$ encode the history of the computation transforming a certain code into another
\item \emph{terms} $M,N$ are lambda expressions that have been
  partially evaluated: they are similar to codes, but their bangs are
  annotated with trails.
\end{itemize}

The \emph{types} of the language (denoted by metavariables
$A, B, C, \ldots$) include atomic types $P$, functions $A \imp B$, and
\emph{audited units} $\AU{s} A$. Notice that types contain codes,
which can contain types in turn.  Audited types correspond to terms
that capture and record all the computation happening inside them, for
subsequent retrieval.


The syntax of the language is presented in
Figure~\ref{fig:syntax}. Codes and terms are defined in a similar way
to Pfenning and Davies~\cite{Pfenning2001}. We use the symbol
`$\bang$' (``bang'') to denote box introduction in both codes and
terms; the corresponding elimination form is given by the $\tlet$
operator, which unpacks an audited code or term allowing us to access
its contents. The variable declared by a $\tlet$ is bound in its
second argument: in essence, $\tlet(u := \bang_q M,N)$ will reduce to
$N$, where $u$ has been replaced by $M$; $q$ will be used to produce a
trail explaining this reduction.

\begin{example}\label{ex:letbang}
$\JLr$ allows us to manipulate history-carrying data explicitly. For instance, we could have the following history-carrying natural numbers:
\[
\bang_q 2 : \AU{1+1} \bbN 
\qquad \qquad
\bang_{q'} 6 : \AU{\mathsf{fact}~3} \bbN
\]
The trail $q$ represents the history of $2$, i.e. a witness to the computation that produced $2$ by adding $1+1$; similarly, $q'$ describes how computing $\mathsf{fact}~3$ (for a suitable definition of the factorial function) yielded $6$.

Now supposing we wish to add these two numbers together, at the same time retaining their history, we will use the $\tlet$ construct to look inside them:
\[
\bang~\tlet(u^\bbN := \bang_q 2, \tlet(v^\bbN := \bang_{q'} 6, u + v))
\stackrel{*}{\red} \bang_{q''} 8 : \AU{\tlet(u^\bbN := \bang~1+1, \tlet(v^\bbN := \bang~\mathsf{fact}~3, u+v))}~\bbN
\]
where the type of the terms is preserved under reduction, and the final trail $q''$, produced by composing $q$ and $q'$, is expected to log the reductions of $1+1$, of $\mathsf{fact}~3$, and of the two $\tlet$s.
\end{example}

Trail inspection codes and terms will perform computation by primitive
recursion on a certain trail, which represents an audited unit's
computation history. Trail inspection references the computation
history of the current context ($\tinsp(\theta)$ and
$\tinsp(\vartheta)$).  
Here, $\theta$ and
$\vartheta$ contain the branches that define the recursion.
\begin{figure*}[t]
{\small{
\begin{center}
\begin{tabular}{rclrcll}
\multicolumn{3}{l}{Codes:} & \multicolumn{3}{l}{Terms:} & \\
$s, t, \ldots$ & $::=$ & $a$ & $M, N, \ldots$ & $::=$ & $a$ & simple variable \\
& $|$ & $u$ & & $|$ & $u$ & audited variable \\
& $|$ & $\lambda a^A.s$ & & $|$ & $\lambda a^A.M$ & abstraction \\
& $|$ & $(s \; t)$ & & $|$ & $(M \; N)$ & application \\
& $|$ & $\bang s$ & & $|$ & $\bang_q M$ & box introduction \\
& $|$ & $\tlet(u^A := s, t)$ & & $|$ & $\tlet(u^A := M, N)$ & box elimination \\
& $|$ & $\tinsp(\theta)$ & & $|$ & $\tinsp(\vartheta)$ & trail inspection 
\end{tabular}

\medskip

\begin{tabular}{rcllrcll}
\multicolumn{4}{l}{\emph{Trails:}} &
\multicolumn{4}{l}{\emph{Trail inspection branches:}} \\
$q, q', \ldots$ & $::=$ & $\qrefl(s)$ & reflexivity &
$\theta, \theta', \ldots$ & $::=$ & $\{\vect{s/\psi}\}$ & for codes \\
& $|$ & $\qtrans(q,q')$ & transitivity &
$\vartheta, \vartheta', \ldots$ & $::=$ & $\{\vect{M/\psi}\}$ & for terms \\
& $|$ & $\qba(a^A.s,t)$ & beta reduction &
$\zeta, \zeta', \ldots$ & $::=$ & $\{\vect{q/\psi}\}$ & for trails \\
& $|$ & $\qbb(s, u^A.t)$ & box beta reduction &
\multicolumn{4}{l}{\emph{Trail inspection branch labels:}} \\
& $|$ & $\qoti(q,\theta)$ & trail inspection reduction &
$\psi, \psi', \ldots$ & $::=$ & \multicolumn{2}{l}{$\qrefl \orelse \qtrans \orelse \qba \orelse \qbb \orelse \qoti $} \\
& $|$ & $\qlam(a^A.q)$ & congruence wrt lambda &
 & $\orelse$ & \multicolumn{2}{l}{$ \qlam \orelse\qapp \orelse \qlet $} \\
& $|$ & $\qapp(q,q')$ & congruence wrt application &
 & $\orelse$ & \multicolumn{2}{l}{$\qotr_{[]} \orelse \qotr_{::}\orelse \qd$} \\
& $|$ & $\qlet(q,u^A.q')$ & congruence wrt let &
& & & \\
& $|$ & $\qotr(\zeta)$ & congruence wrt trail insp. &
\end{tabular}
\end{center}
}}
\caption{Syntax of \JLr}\label{fig:syntax}
\end{figure*}

 The only difference between codes and terms is that in terms, bang
 operators are annotated with a trail $q$: this follows the intuition
 that terms are codes considered with their computational
 evolution. Trails represent (multi-step) reductions relating two
 codes: thus, their structure matches the derivation tree of the
 reduction itself. 

 In this definition, when a trail contains an argument composed of a
 variable and another expression separated by a dot, the variable
 should be considered bound in that expression. The constructors
 $\qba, \qbb$ and $\qoti$ define atomic reductions: $\qba$ represents
 the regular $\beta$-reduction; $\qbb$ is similar, but reduces a
 let-bang combination; finally, $\qoti$ reduces trail inspections. The
 constructors $\qlam$, $\qapp$, $\qlet$, and $\qotr$ provide
 congruence rules for all code constructors (with the notable
 exception of bangs, since trails are not allowed to escape an audited
 unit): for example, if $q$ is a trail from $s$ to $t$, $\qlam(a^A.q)$
 constructs a trail from $\lambda a^A.s$ to $\lambda a^A.t$. Finally,
 the definition of trails is completed by reflexivity and transitivity
 (but excluding symmetry). Notice that reflexive trails take as a parameter
 the code whose identity is stated (for brevity, we will keep this
 parameter implicit in our examples).


\begin{example}\label{ex:applambda}
We can build a pair of natural numbers using Church's encoding: \JLr makes it possible to keep a log of this operation without adding specialized code:
\begin{align*}
& \bang~((\lambda x,y,p.p~x~y)~2)~6 \\
\red & \bang_{\qtrans(\qrefl,\qapp(\qba(x.\lambda y,p.p~x~y,2),\qrefl))}~(\lambda y,p.p~2~y)~6 \\
\red & \bang_{\qtrans(\qtrans(\qrefl,\qapp(\qba(x.\lambda y,p.p~x~y,2),\qrefl)),\qba(y.\lambda p.p~2~y,6))}~\lambda p.p~2~6
\end{align*}
The trail for the first computation step is obtained by transitivity (trail costructor $\qtrans$) from the original trivial trail ($\qrefl$, i.e. reflexivity) composed with $\qba(x.\lambda y,p.p~x~y,2)$, which describes the beta-reduction of $(\lambda x,y,p.p~x~y)~2$: this subtrail is wrapped in a congruence $\qapp$ because the reduction takes place deep inside the left-hand subterm of an application (the other argument of $\qapp$ is reflexivity, because no reduction takes place in the right-hand subterm).

The second beta-reduction happens at the top level and is thus not wrapped in a congruence. It is combined with the previous trail by means of transitivity.
\end{example}

Inspection branches $\theta$, $\vartheta$ and $\zeta$ are maps from trail constructors to codes, terms, or trails. More precisely, since some of the trail constructors have a variable arity, in that case we need to distinguish a nil case ($[]$) and a cons case ($::$) to specify the corresponding branch of recursion; in addition, we do not require the map to be exhaustive, but allow a ``default'', catch-all branch specified using the label $\qd$: this does not make the language more expressive, but will allow us to extend the language with additional constants without needing to modify the base syntax.
We now give a formal definition of the result of trail inspection.
\begin{definition}
The operation $q\theta$, which produces a code by structural recursion on $q$ applying the inspection branches $\theta$, is defined as follows:
{\small{
\begin{align*}
\psi(\_)\theta &\triangleq \theta(\qd) \qquad \mbox{(if $\psi \notin \dom(\theta)$)} 
& 
\qlam(a^A.q)\theta &\triangleq \theta(\qlam)~(q\theta)
\\
\qrefl(\_)\theta &\triangleq \theta(\qrefl) 
&
\qapp(q,q')\theta &\triangleq \theta(\qapp)~(q\theta)~(q'\theta)
\\
\qtrans(q,q')\theta & \triangleq
  \theta(\qtrans)~(q\theta)~(q'\theta)
&
\qlet(q,u^A.q')\theta &\triangleq \theta(\qlet)~(q\theta)~(q'\theta) 
\\
\qba(\_,\_)\theta &\triangleq \theta(\qba) 
&
\qotr(\{ \})\theta &\triangleq \theta(\qotr_{[]})
\\
\qbb(\_,\_)\theta &\triangleq \theta(\qbb)
&
\qotr(\{q/\psi,\vect{q'/\psi'}\})\theta &\triangleq 
  \theta(\qotr_{::})~(q\theta)~(\qotr(\{\vect{q'/\psi'}\})\theta)
\\
\qoti(\_,\_)\theta &\triangleq \theta(\qoti) 
&
\end{align*}
}}
The twin operation $q\vartheta$, which produces a term by structural recursion on $q$, is defined similarly, by replacing each occurrence of $\theta$ with $\vartheta$.
\end{definition}

\begin{example}\label{ex:profiling}
For profiling purposes, we might be interested in counting the number of computation steps that have happened thus far in an audited unit. This can be accomplished by means of a trail replacement $\vartheta_+$ such that:
\[
\vartheta_+(\psi) = 
  \begin{cases}
  1 & \text{if $\psi = \qba, \qbb, \qoti$} \\
  \lambda x.x & \text{if $\psi = \qlam$} \\
  \lambda x,y. x+y & \text{if $\psi = \qtrans, \qapp, \qlet, \qotr_{::}$} \\
  0 & \text{else}
  \end{cases}
\]
This trail replacement counts $1$ for each reduction step ($\qba, \qbb, \qoti$); the cases for transitivity and congruences add together the recursive results for subtrails; the last clause ignores other constructors, most notably $\qrefl$ and $\qoti_{[]}$, which are not reduction steps and do not have subtrails.
\end{example}

\subsection{Substitutions}
\JLr{} employs two notions of substitution, one for each kind of variable supported by the language:
\begin{itemize}
\item Simple variable substitution, which is not history-aware, is
  applied to codes, types, or terms and maps a simple variable to,
  respectively, a code or a term; we denote it by $\subst{a}{t}$ or
  $\subst{a}{N}$, which can also be written $\gamma$ for short.
\item Audited variable substitution, which is history-aware, is
  applied to codes or terms and maps an audited variable to,
  respectively, a code $t$ or a tuple $(N,q,t)$; we denote it by
  $\subst{u}{t}$ or $\subst{u}{(N,q,t)}$, which can also be written
  $\delta$ for short.
\end{itemize}
We will also use the letter $\eta$ to refer to either of the two flavors of substitution when more precision is not needed.

\begin{definition}[substitution on codes]\label{defn:csubst}
Substitution
of simple and audited variables on a code $r$ ($r[s/a]$ and
$r[s/u]$) is defined as usual (avoiding variable capture). The full
definition is shown in Appendix~\ref{app:codesubst}.
\end{definition}

Each of the two substitutions on codes is extended to trails and types
in the
obvious way, by traversing the trail or type until a subexpression containing
a code $s$ is found, then resorting to substitution on $s$. 
To define substitution on terms, we first need auxiliary
definitions of the \emph{source} and \emph{target} of a trail:
\begin{definition}\label{defn:srctgt}
The source and target of a trail $q$ ($\src(q)$ and $\tgt(q)$) are defined as follows:
{\small{
\begin{align*}
\src(\qrefl(s)) & \triangleq s 
&
\tgt(\qrefl(s)) & \triangleq s
\\
\src(\qtrans(q_1,q_2)) & \triangleq \src(q_1)
&
\tgt(\qtrans(q_1,q_2)) & \triangleq \tgt(q_2)
\\ 
\src(\qba(a^A.s_1,s_2)) & \triangleq (\lambda a^A.s_1)~s_2
&
\tgt(\qba(a^A.s_1,s_2)) & \triangleq s_1\subst{a}{s_2}
\\
\src(\qbb(s_1,v^A.s_2)) & \triangleq \tlet(v^A := \bang s_1, s_2) 
&
\tgt(\qbb(s_1,v^A.s_2)) & \triangleq s_2\subst{v}{s_1}
\\
\src(\qoti(q',\theta)) & \triangleq \tinsp(\theta)
&
\tgt(\qoti(q',\theta)) & \triangleq q'\theta
\end{align*}
}}
The omitted cases are routine and shown in Appendix~\ref{app:srctgt}.
\end{definition}
A valid trail $q$ represents evidence that $\src(q)$ reduces to $\tgt(q)$ (we formally state this result in the next section).

We omit the definition of simple variable substitution on terms for brevity (it is similar to the corresponding substitution on codes).
The definition of audited variable substitution on terms, however, deserves some more explanation: since this substitution arises from the unpacking of a bang term, which contains a trail we need to take care of, it takes the form $ \subst{u}{(N,q,t)}$, where $q$ is the computation history that led from the code $t$ to the term $N$.

For related reasons, this substitution can be applied to a term to produce both a new term and a corresponding trail: the two operations are written $M\times\delta$ and $M\ltimes\delta$, respectively. 

\begin{definition}[audited variable substitution (terms)]\label{defn:tmvsubst}
The operations $M \times \subst{u}{(N,q,t)}$ and $M \ltimes \subst{u}{(N,q,t)}$ are defined as follows:
\[\small
\begin{array}{rclcrcl}
u \ltimes\subst{u}{(N,q,t)} & \triangleq& N&&
 (\bang_{q'} R) \ltimes \subst{u}{(N,q,t)} & \triangleq&
                                                        \bang_{\qtrans(q'\subst{u}{t},R \times\subst{u}{(N,q,t)})} (R \ltimes \subst{u}{(N,q,t)})\\ 
u \times\subst{u}{(N,q,t)} &\triangleq& q &&
(\bang_{q'} R) \times \subst{u}{(N,q,t)} &\triangleq& \qrefl(\bang (\src(q')\subst{u}{t}))
\end{array}\]
The omitted cases are routine and shown in
Appendix~\ref{app:avarsubst}. The main feature of interest here is that when we $\ltimes$-substitute
$(N,q,t)$ for $u$ in an audited unit $\bang_{q'}$, we need to augment
the trail with a step showing how the replaced occurrence of $u$
evaluated from $t$ to $N$, as well as substituting for other
occurrences of $u$; the corresponding case of $\times$-substitution returns the appropriate reflexivity trail, since no trail escapes a bang as a result of substitution.
\end{definition}

\begin{example}\label{ex:letbang2}
We consider again the term from Example~\ref{ex:letbang}:
\[
\bang~\tlet(u^\bbN := \bang_q 2, \tlet(v^\bbN := \bang_{q'} 6, u + v))
\]
Audited variable substitution is used to reduce $\tlet$-$\bang$ combinations; for example, to reduce the outer $\tlet$ in the example, we need to compute
\begin{align*}
\tlet(v^\bbN := \bang_{q'} 6, u + v)\lsubst{u}{2,q,\src(q)} & = \tlet(v^\bbN :=   \bang_{q'} 6, 2 + v) \\
\tlet(v^\bbN := \bang_{q'} 6, u + v)\csubst{u}{2,q,\src(q)} & = \qlet(\qrefl, \qapp(\qapp(\qrefl,q),\qrefl))
\end{align*}
(we assume that the infix notation for $+$ is syntactic sugar for two nested applications).
\end{example}

\subsection{Type system}
We now introduce typing rules for codes, trails, and terms. These rules use three corresponding judgments with the following shape:
\begin{center}
\begin{tabular}{ll}
  $\qchk{\CTX}{s}{A}$ & $s$ has type $A$ \\
  $\qchk{\CTX}{q}{s \direq_A t}$ & $q$ witnesses  $s$ reducing to $t$, with type $A$ \\
  $\tchk{\CTX}{M}{A}{s}$ &$M$ has type $A$ and code $s$
\end{tabular}
\end{center}
The environments $\Delta$ and $\Gamma$ are list of type declarations
for audited and simple variables respectively, in the form $u :: A$ or
$a : A$.

\begin{figure*}[p]
{\small{
\begin{center}
\begin{tabular}{c}
\multicolumn{1}{l}{\fbox{$\qchk{\CTX}{s}{A}$ }}
\\\rVS
 \AxiomC{$a:A \in \Gamma$} \RightLabel{\textsc{Var}}
 \UnaryInfC{$\qchk{\CTX}{a}{A}$}
 \DisplayProof
 \qquad
 \AxiomC{$\qchk{\Delta;\Gamma,a:A}{s}{B}$} \RightLabel{$\imp \Intro$}
 \UnaryInfC{$\qchk{\CTX}{\lambda a^A.s}{A \imp B}$}
 \DisplayProof
\qquad 
 \AxiomC{$\qchk{\CTX}{s}{A \imp B}$} 
 \AxiomC{$\qchk{\CTX}{t}{A}$}
  \RightLabel{$\imp \Elim$}
 \BinaryInfC{$\qchk{\CTX}{s\;t}{B}$}
 \DisplayProof
 \\ \rVS
 \AxiomC{$u::A \in \Delta$} \RightLabel{\textsc{mVar}}
 \UnaryInfC{$\qchk{\Delta;\Gamma}{u}{A}$}
 \DisplayProof
  \qquad
 \AxiomC{$\qchk{\Delta;\cdot}{t}{A}$} \RightLabel{$\Box \Intro$}
 \UnaryInfC{$\qchk{\Delta;\Gamma}{!t}{\AU{t} A}$}
 \DisplayProof
 \qquad
 \AxiomC{$\qchk{\CTX}{s}{\AU{r} A}$} 
 \AxiomC{$\qchk{\Delta,u::A;\Gamma}{t}{C}$} 
   \RightLabel{$\Box \Elim$}
 \BinaryInfC{$\qchk{\CTX}{\tlet(u^A := s,t)}{C\subst{u}{r}}$}
 \DisplayProof
\\ \rVS
 \AxiomC{$\qd \in \dom(\theta)$}
 \AxiomC{$\left[\qchk{\Delta;\cdot}{\theta(\psi)}{\cT^B(\psi)}\right]_{\psi \in \dom(\theta)}$} 
   \RightLabel{\textsc{TI}}
 \BinaryInfC{$\qchk{\CTX}{\tinsp(\theta)}{B}$} 
 \DisplayProof
\\ 
\multicolumn{1}{l}{\fbox{ $\tchk{\CTX}{M}{A}{s}$}}
\\\rVS
 \AxiomC{$a:A \in \Gamma$} \RightLabel{\textsc{T-Var}}
 \UnaryInfC{$\tchk{\CTX}{a}{A}{a}$}
 \DisplayProof
 \qquad
 \AxiomC{$\tchk{\Delta;\Gamma,a:A}{M}{B}{s}$} \RightLabel{\textsc{T-Abs}}
 \UnaryInfC{$\tchk{\CTX}{\lambda a^A.M}{A \imp B}{\lambda a^A.s}$}
 \DisplayProof
\\\rVS
 \AxiomC{$\tchk{\Delta;\Gamma}{M}{A \imp B}{s}$} 
 \AxiomC{$\tchk{\Delta;\Gamma}{N}{A}{t}$}
  \RightLabel{\textsc{T-App}}
 \BinaryInfC{$\tchk{\Delta;\Gamma}{M\;N}{B}{s\;t}$}
 \DisplayProof
 \\ \rVS
 \AxiomC{$u::A \in \Delta$} \RightLabel{\textsc{T-mVar}}
 \UnaryInfC{$\tchk{\CTX}{u}{A}{u}$}
 \DisplayProof
 \qquad
 \AxiomC{$\tchk{\Delta;\cdot}{M}{A}{t}$} 
 \AxiomC{$\qchk{\Delta;\cdot}{q}{s \direq_A t}$}  
 \RightLabel{\textsc{T-Bang}}
 \BinaryInfC{$\tchk{\CTX}{\bang_q M}{\AU{s} A}{\bang s}$}
 \DisplayProof
 \\\rVS
 \AxiomC{$\tchk{\CTX}{M}{\AU{r} A}{s}$} 
 \AxiomC{$\tchk{\Delta,u::A;\Gamma}{N}{C}{t}$} 
   \RightLabel{\textsc{T-Let}}
 \BinaryInfC{$\tchk{\CTX}{\tlet(u^A := M,N)}{C\subst{u}{r}}{\tlet(u^A := s,t)}$}
 \DisplayProof
\\
 \rVS
 \AxiomC{$\qd \in \dom(\vartheta)$} 
 \AxiomC{$\dom(\vartheta) = \dom(\theta)$} 
 \AxiomC{$\left[\tchk{\Delta;\cdot}{\vartheta(\psi)}{\cT^B(\psi)}{\theta(\psi)}\right]_{\psi \in \dom(\vartheta)}$} 
   \RightLabel{\textsc{T-TI}}
 \TrinaryInfC{$\tchk{\CTX}{\tinsp(\vartheta)}{B}{\tinsp(\theta)}$} 
 \DisplayProof
\\ 
\multicolumn{1}{l}{\fbox{ $\qchk{\CTX}{q}{s \direq_A t}$}}
\\
 \rVS
 \AxiomC{$\qchk{\CTX}{s}{A}$} \RightLabel{\textsc{Q-Refl}}
 \UnaryInfC{$\qchk{\CTX}{\qrefl(s)}{s \direq_A s}$}
 \DisplayProof
\hspace{.5cm}
 \AxiomC{$\qchk{\CTX}{q_1}{r \direq_A s}$} 
 \AxiomC{$\qchk{\CTX}{q_2}{s \direq_A t}$} 
   \RightLabel{\textsc{Q-Trans}}
 \BinaryInfC{$\qchk{\CTX}{\qtrans(q_1,q_2)}{r \direq_A t}$}
 \DisplayProof
\\
 \rVS
 \AxiomC{$\qchk{\Delta;\Gamma,a:A}{s}{ B}$} 
 \AxiomC{$\qchk{\CTX}{t}{A}$}
   \RightLabel{\textsc{Q-$\qba$}}
 \BinaryInfC{$\qchk{\CTX}{\qba(a^A.s,t)}{(\lambda a^A.s) \; t \direq_B s\subst{a}{t}}$}
 \DisplayProof
\\\rVS
 \AxiomC{$\qchk{\Delta;\cdot}{s}{A}$}
 \AxiomC{$\qchk{\Delta,u::A;\Gamma}{t}{C}$}
   \RightLabel{\textsc{Q-$\qbb$}}
 \BinaryInfC{$\qchk{\Delta;\Gamma}
              {\qbb(s, u^A.t)}
              {\tlet(u^A := \bang s,t) \direq_{C\subst{u}{s}} t\subst{u}{s}}$}
 \DisplayProof
\\
 \rVS
 \AxiomC{$\qd \in \dom(\vartheta)$} 
 \AxiomC{$\qchk{\Delta;\cdot}{q}{s \direq_A t}$}
 \AxiomC{$\left[\qchk{\Delta;\cdot}{\theta(\psi)}{\cT^B(\psi)}\right]_{\psi \in \dom(\theta)}$} 
  \RightLabel{\textsc{Q-TI}}
 \TrinaryInfC{$\qchk{\CTX}{\qoti(q,\theta)}
  {\tinsp(\theta) \direq_B q\theta}$}
 \DisplayProof
\\
 \rVS
 \AxiomC{$\qchk{\Delta;\Gamma,a:A}{q}{s \direq_B t}$}
   \RightLabel{\textsc{Q-Abs}}
 \UnaryInfC{$\qchk{\CTX}{\qlam(a^A.q)}{\lambda a^A.s \direq_{A \imp B} \lambda a^A.t}$} 
 \DisplayProof
\\ \rVS
 \AxiomC{$\qchk{\Delta;\Gamma}{q_1}{s_1 \direq_{A \imp B} t_1}$} 
 \AxiomC{$\qchk{\Delta;\Gamma}{q_2}{s_2 \direq_A t_2}$}
   \RightLabel{\textsc{Q-App}}
 \BinaryInfC{$\qchk{\Delta;\Gamma}{\qapp(q_1,q_2)}{s_1 \; s_2 \direq_B t_1 \; t_2}$} 
 \DisplayProof
\\
 \rVS
 \AxiomC{$\qchk{\CTX}{q_1}{s_1 \direq_{\AU{r} A} t_1}$}
 \AxiomC{$\qchk{\Delta, u::A;\Gamma}{q_2}{s_2 \direq_{C} t_2}$}
   \RightLabel{\textsc{Q-Let}}
 \BinaryInfC{$\qchk{\CTX}{\qlet(q_1,u^A.q_2)}{\tlet(u^A := s_1,s_2) 
    \direq_{C\subst{u}{r}}
    \tlet(u^A := t_1,t_2)}$} 
 \DisplayProof
\\
 \rVS
 \AxiomC{$\qd \in \dom(\zeta)$} 
 \AxiomC{$\dom(\zeta) = \dom(\theta) = \dom(\theta')$} 
 \AxiomC{$\left[\qchk{\Delta;\cdot}{\zeta(\psi)}{\theta(\psi) \direq_{\cT^B(\psi)} \theta'(\psi)}\right]_{\psi \in \dom(\zeta)}$}
   \RightLabel{\textsc{Q-Trpl}}
 \TrinaryInfC{$\qchk{\CTX}{\qotr(\zeta)}{\tinsp(\theta) \direq_B \tinsp(\theta')}$} 
 \DisplayProof
\end{tabular}
\end{center}
}}
\caption{Typing rules for \JLr}\label{fig:rules}
\end{figure*}

Figure~\ref{fig:rules} shows the typing rules of the language. In the rules for codes, if we forget for a moment about trail inspections, it is easy to recognize Pfenning and Davies's judgmental presentation of modal logic. 
The rules for terms are very similar, the biggest difference being found in rule \textsc{T-Bang}: to typecheck $\bang_q M$, we first obtain the type $A$ and code $t$ for $M$; then we typecheck the trail $q$ to ensure that it has type $A$ and target $t$; if both checks are successful, we can return $\AU{s} A$ as the type of the expression, and $\bang s$ as its code, where $s$ is the source of $q$. This also shows that when we say that $\bang s$ is the code of $\bang_q M$, we mean that the term is the result of reducing the code, and the trail $q$ records the computation history.

The rules for trail inspection, both as a code and as a term, rely on an auxiliary definition of $\cT^B$, which is a function from inspection branch labels to the corresponding type; it depends on the superscript $B$, which refers to the output type of the inspection. 
The rule for inspection codes (\textsc{TI}) checks that the
types of all the branches in $\theta$ match their labels. 
This requires $\theta$ to be defined on the default branch $\qd$ to
guarantee that the inspection is exhaustively defined.
We define $\cT^B$ as follows:
\[
\cT^B(\psi) \triangleq \begin{cases}
  B & (\psi = \qd, \qrefl, \qba, \qbb, \qoti, \qiti, \qotr_{[]}) \\
  B \imp B & (\psi = \qlam, \qitr_{[]}) \\
  B \imp B \imp B & (\psi = \qtrans, \qapp, \qlet, \qitr_{::}, 
    \qotr_{::}) 
  \end{cases}
\]

Rather than typecheck a term to compute its code, we can define a function performing this operation directly.
\begin{definition}\label{defn:cdoftm}
The code of a term $M$ (notation: $\code(M)$) is the code obtained by
replacing all occurrences of $\bang_q N$ with $\bang\src(q)$.  The
full definition is shown in Appendix~\ref{app:codes}.
\end{definition}

The rule \textsc{T-Bang}, which typechecks audited terms, needs to
compute the type of a trail: this is performed by the typechecking
rules for trails. Each of these rules typechecks a different trail
constructor: in particular, four rules define the contraction of
redexes (\textsc{Q-$\qba$}, \textsc{Q-$\qbb$}, \textsc{Q-TI}),
\textsc{Q-Refl} and \textsc{Q-Trans} ensure that trails induce a
preorder on codes, and the remaining rules model congruence with
respect to all code constructors but bangs. Let us remark that the
trail $q$ mentioned in trail inspections cannot be synthesized from
the redex (since codes are not history-aware), but must be provided as
an argument to $\qoti$.

We can prove that operations of Definitions~\ref{defn:srctgt} and~\ref{defn:cdoftm} match the typing judgments:
\begin{lemma}
If $\qchk{\CTX}{q}{s \direq_A t}$ then $\src(q) = s$ and $\tgt(q) = t$.
\end{lemma}
\begin{lemma}
If $\tchk{\CTX}{M}{A}{s}$ then $\code(M) = s$.
\end{lemma}

\subsection{Semantics}
In this section, we define a reduction relation expressing how terms compute to values. Reduction differs from trails since it relates terms (rather than codes), and since unlike trails it does not appear as part of terms (however, it is defined in such a way that reduced terms will contain modified trails expressing a \emph{log} of reduction).

Our definition of reduction makes use of contexts, i.e. terms with holes, represented by black boxes ($\blacksquare$); similarly, we provide a notion of trail contexts, denoted by $\cQ$.
The notation $\cE[M]$ indicates the term obtained by filling the hole in $\cE$ with $M$; $\cQ[q]$ is defined similarly.
\[\small
\begin{array}{rcl}
\cE & ::= & \blacksquare \orelse \lambda a^A.\cE \orelse (\cE \; M)
            \orelse (M \; \cE) \orelse \bang_q \cE \orelse \tlet(u^A := \cE, M) \orelse
            \tlet(u^A := M, \cE) \orelse\tinsp(\{ \vect{M/\psi}, \cE/\psi', \vect{N/\psi''} \}) \smallskip\\
\cQ & ::= & \blacksquare \orelse \qtrans(\cQ,q) \orelse \qtrans(q,\cQ) \orelse \qapp(\cQ,q) \orelse \qapp(q,\cQ) \orelse \qlam(a^A.\cQ) \orelse \qlet(\cQ,u^A.q) \orelse \qlet(q,u^A.\cQ) \\
& | & \qotr(\{ \vect{q_1/\psi}, \cQ/\psi', \vect{q_2/\psi''} \}) 
\end{array}
\]

Following~\cite{Bavera2015} we use $\cF$ to refer to contexts that do not allow holes inside bangs.
Box-free contexts $\cF$ and trail contexts $\cQ$ are related by the following definition:
\begin{definition}[canonical trail context]
For every $\cF$, there exists a canonical trail context $\cQ_\cF$, defined as follows:
\[\small{
\begin{array}{rclcrcl}
\cQ_\blacksquare & \triangleq& \blacksquare
&&
\cQ_{(\cF \; M)} & \triangleq& \qapp(\cQ_\cF,\qrefl(\code(M))) \\
\cQ_{(M \; \cF)} & \triangleq& \qapp(\qrefl(\code(M)),\cQ_\cF)
&&
\cQ_{\lambda a^A.\cF} & \triangleq &\qlam(a^A.\cQ_\cF) \\
\cQ_{\tlet(\cF,u^A.M)} & \triangleq &\qlet(\cQ_\cF,u.\qrefl(\code(M)))
&&
\cQ_{\tlet(M,u^A.\cF)} & \triangleq& \qlet(\qrefl(\code(M)),u^A.\cQ_\cF)
\smallskip\\
\cQ_{\tinsp(\{ \vect{M/\psi}, \cF/\psi', \vect{N/\psi''} \})} & \triangleq &\multicolumn{4}{l}{ \qotr(\{ \vect{\qrefl(\code(M))/\psi},Q_\cF/\psi'. \vect{\qrefl(\code(N))/\psi''} \})}
\end{array}
}\]
\end{definition}
Informally, $\cQ_\cF$ uses congruences to express a reflexive trail for $\cF$, with a hole as the subtrail corresponding to the hole in $\cF$. It is $\cQ_\cF$ that is responsible for producing the $\qapp$ congruence of Example~\ref{ex:applambda}.

Thanks to trail contexts and our definition of audited variable
substitution, we can define reduction directly, without an auxiliary
judgment pushing trails towards the closest outer bang. To avoid
dealing with the unwanted situation of trail inspections not guarded
by a bang, we only consider terms surrounded by an outer bang.

\begin{figure*}[tb]
{\small{
\begin{center}
\begin{tabular}{c}
\rVS
\AxiomC{\phantom{$A$}}\RightLabel{\textsc{B-$\beta$}}
\UnaryInfC{$\bang_q \cF[(\lambda a^A.M) \; N)] \red \bang_{\qtrans(q,\cQ_\cF[\qba(a.\code(M),\code(N))])} \cF[M\subst{a}{N}]$} 
\DisplayProof
\quad
\AxiomC{$M \red N$}\RightLabel{\textsc{B-Bang}}
\UnaryInfC{$\bang_q \cF[M] \red \bang_q \cF[N]$}
\DisplayProof
\\
\rVS
\AxiomC{$q_f = \qtrans(q,\cQ_\cF[\qtrans(\qbb(\src(q'),\code(N)),N\times\subst{u}{(M,q',\src(q'))})])$} \RightLabel{\textsc{B-$\beta_\Box$}}
\UnaryInfC{$\bang_q \cF[\tlet(\bang_{q'} M,u.N)] \red \bang_{q_f} \cF[N\ltimes\subst{u}{(M,q',\src(q'))}]$}
\DisplayProof
\\
\rVS
\AxiomC{\phantom{$A$}} \RightLabel{\textsc{B-TI}}
\UnaryInfC{$\bang_q \cF[\tinsp(\vartheta)] \red \bang_{\qtrans(q,\cQ_\cF[\qoti(q,\code(\vartheta))])} \cF[q\vartheta]$}
\DisplayProof
\end{tabular}
\end{center}
}}
\caption{Term reduction rules for \JLr}\label{fig:redrules}
\end{figure*}

The reduction rules are defined in Figure~\ref{fig:redrules}. They
operate by contracting a subterm appearing in a context $\cF$, at the
same time updating the trail of the enclosing bang by asserting that
$q$ is followed by a new contraction appearing in the trail context
$\cQ_\cF$. Rule \textsc{B-$\beta_\Box$} is an exception, in that after
performing the $\qbb$ contraction, we still need to take into account
the residuals of $q'$, expressed by the substitution
$N \csubst{u}{(M,q',\src(q'))}$

We write $s \leadsto t$ when there exists a well-typed $q$ such that $\src(q) = s$ and $\tgt(q)= t$.

\begin{example}
We take again the term from Example~\ref{ex:letbang} and reduce the outer $\tlet$ as follows:
\[
\bang~\tlet(u^\bbN := \bang_q 2, \tlet(v^\bbN := \bang_{q'} 6, u + v))
\red
\bang_{q_f}~\tlet(v^\bbN := \bang_{q'} 6, 2 + v)
\]
where we use the substitutions we computed in Example~\ref{ex:letbang2}. Since the reduction happens immediately inside the bang, we use rule \textsc{B-$\qbb$} with $\cQ_\blacksquare = \blacksquare$, thus $q_f$ is as follows:
\begin{align*}
q_f & = \qtrans(\qrefl, \qtrans(\qbb(\src(q),u.\code(\tlet(v := \bang_{q'} 6, u + v)),\qlet(\qrefl, \qapp(\qapp(\qrefl,q),\qrefl)))) \\
& = \qtrans(\qrefl,\qtrans(\qbb(1+1,u.\tlet(v := \bang~\mathsf{fact}~3, u + v)),\qlet(\qrefl, \qapp(\qapp(\qrefl,q),\qrefl))))
\end{align*}
(where $\src$ and $\code$ compute according to the hypotheses we made about $q$ and $q'$).
\end{example}

\begin{example}
To show how to use the profiling inspection from Example~\ref{ex:profiling}, we need to assume a certain evaluation strategy: for example, call-by-value.
We then write $(x \leftarrow M; N)$ as syntactic sugar for $(\lambda x.N)~M$ and evaluate the following term:
\[
\begin{array}{l}
!~(t_0 \leftarrow \tinsp(\vartheta_+); \\
   \quad \_ \leftarrow ((\lambda x,y,p.p~x~y)~2)~6; \\
   \quad t_1 \leftarrow \tinsp(\vartheta_+); \\
   \quad t_1 - t_0)
\end{array}
\]
In this term, $t_0$ evaluates to $0$ because the bang starts with a reflexive trail; by the time we get to the third line, the outer trail contains a log of the inspection on the first line, the two beta reductions needed to evaluate the second line, and two more beta reductions to account for sequential compositions: thus $t_1$ evaluates to $5$; finally, we evaluate $t_1 - t_0 = 5$.
\end{example}

\begin{definition}[normal form]
A term $M$ is in normal form iff for all trails $q$ and all terms $N$,
$\bang_q M \not\leadsto N$.  Likewise, a code is in normal form iff
there exists no $t \neq s$ such that $s \leadsto t$.
\end{definition}

It should be noted that \JLr, despite being strongly normalizing (as we will prove in the next section), is not confluent: the same term may reduce to different values under different reduction strategies. In particular, the
trails appearing in bangs are sensitive to the reduction order: a call-by-value strategy and a call-by-name strategy
will produce different trails. We could recover confluence by forcing
a certain evaluation strategy and quotienting trails by means of a
suitable equivalence relation: this is beyond the scope of the present paper.

The main properties of \JLr, such as subject reduction, can be proved
similarly to those for \JLh.  As explained in Section~\ref{sec:intro},
we also need to establish the consistency of trails as proofs of
reducibility. 
\begin{theorem}
For all codes $s,t$ of \JLr in normal form such that $s \neq t$, there exist no $\Delta, \Gamma, q, A$ such that $\qchk{\CTX}{q}{s \direq_A t}$
\end{theorem}
\begin{proof}
If such a judgment were provable, by structural induction on it we would be able to show that $q$ must be a combination of reflexivity, transitivity, and congruence rules (since both $s$ and $t$ are in normal form, no contraction is possible in absence of symmetry). Then $s = t$, which falsifies the hypothesis.
\end{proof}

\section{Strong Normalization}\label{sec:sn}
We now summarize our proof of strong normalization for \JLr{}. The presence of a recursion operator on trails, whose occurrences can be arbitrarily nested, together with the fact that trails grow monotonically during execution, makes this result non-trivial.

 Our proof employs the well-known notion of ``candidates of reducibility'' \cite{GirardCR} (sets of strongly normalizing terms enjoying certain desirable properties). Candidates of reducibility are a powerful and flexible tool, which has been used to prove the strong normalization property of very expressive lambda calculi~\cite{Geuvers94, Luo90, Werner94} and also other results such as the Church-Rosser property~\cite{Gallier90}. 

In the literature, it is possible to find several definitions of candidate of reducibility: along with Girard's definition, we can cite Tait's saturated sets~\cite{Tait75} and Parigot's inductive definition~\cite{Parigot94}.
Our proof employs Girard's candidates, and can be easily compared to the similar proof for the System F~\cite{proofstypes}. Some acquaintance with that proof will be assumed in the rest of this section.
We will use $\SN$ to denote the set of all strongly normalizing terms. Reduction on strongly normalizing terms is a well-founded relation, which allows us to reason by well-founded induction on it.

\subsection{A simplified calculus}
As a technical means to investigate normalization of \JLr, we define
\JLs, a simplified version of the calculus which forgets about trails
associated with box introductions. Its types, terms, and contexts are defined by the following grammar:
\[\small\begin{array}{rcl}
\tau, \sigma & ::= & P \mid \tau \imp \sigma \mid \Box \tau\\
s, t, \ldots & ::= &a 
\orelse u \orelse \lambda a^\tau.s 
\orelse (s \; t)
\orelse \bang s
\orelse \tlet(u^\tau := s, t)
\orelse \tinsp(\theta)\\
\theta, \theta', \ldots & ::=& \{\vect{s/\psi}\}\\
\cE & ::= & \blacksquare \orelse \lambda a^\tau.\cE \orelse (\cE \; s)
            \orelse (s \; \cE) \orelse \bang \cE \orelse \tlet(u^\tau
            := \cE, s) \orelse \tlet(u^\tau := s, \cE)\orelse  \tinsp(\{ \vect{s/\psi}, \cE/\psi', \vect{t/\psi''} \}) 
\end{array}\]
The terms of \JLs{} correspond closely to the codes of \JLr{}. Their semantics, however, is different: $\tinsp(\theta)$ does not perform inspection on the trail of the enclosing bang, but at the time of its evaluation will receive an arbitrary trail. We omit the definition of trails for brevity, but it can be obtained from the corresponding notion in \JLr, by replacing codes with \JLs{} terms.

We can also define simple and audited variable substitution on \JLs{} terms, in a way that mimics the corresponding notion of \JLr (Definition~\ref{defn:csubst}). 

We provide an erasure map from \JLr{} types and terms to \JLs (its extension to contexts is immediate).
\[\small\begin{array}{c}
|P| = P \qquad |A \imp B| = |A| \imp |B| \qquad |\AU{s}A| = \Box |A|
\\
|a| \triangleq a
\qquad
|u| \triangleq u
\qquad
|\lambda a^A.M| \triangleq \lambda a^{|A|}.|M| 
\qquad
|M \; N| \triangleq |M| \; |N|
\\
|\bang_q M| \triangleq \bang |M|
\qquad
|\tlet(u^A := M, N)| \triangleq \tlet(u^A := |M|,|N|)
\qquad
|\tinsp(\vartheta)| \triangleq \tinsp(|\vartheta|)
\end{array}
\]


\begin{figure*}[t]
{\small{
\begin{center}
\begin{tabular}{c}
 \rVS
 \AxiomC{$a:\tau \in \Gamma$}
 \UnaryInfC{$\qchk{\CTX}{a}{\tau}$}
 \DisplayProof
 \quad
 \AxiomC{$u::\tau \in \Delta$}
 \UnaryInfC{$\qchk{\CTX}{u}{\tau}$}
 \DisplayProof
 \quad
 \AxiomC{$\qchk{\Delta;\Gamma,a:\tau}{M}{\sigma}$}
 \UnaryInfC{$\qchk{\CTX}{\lambda a^\tau.M}{\tau \imp \sigma}$}
 \DisplayProof
\quad
 \AxiomC{$\qchk{\Delta;\cdot}{M}{\tau}$}
 \UnaryInfC{$\qchk{\CTX}{\bang M}{\Box \tau}$}
 \DisplayProof
 \\
 \rVS
 \AxiomC{$\qchk{\Delta;\Gamma}{M}{\tau \imp \sigma}$} \noLine
 \UnaryInfC{$\qchk{\Delta;\Gamma}{N}{\tau}$}
 \UnaryInfC{$\qchk{\Delta;\Gamma}{M\;N}{\sigma}$}
 \DisplayProof
 \quad
 \AxiomC{$\qchk{\CTX}{M}{\Box \tau}$} \noLine
 \UnaryInfC{$\qchk{\Delta,u::\tau;\Gamma}{N}{\sigma}$} 
 \UnaryInfC{$\qchk{\CTX}{\tlet(u^\tau := M,N)}{\sigma}$}
 \DisplayProof
 \quad
 \AxiomC{$\qd \in \dom(\theta)$} \noLine
 \UnaryInfC{$\left[\qchk{\Delta;\cdot}{\theta(\psi)}{\cT^\sigma(\psi)}\right]_{\psi \in \dom(\theta)}$} 
 \UnaryInfC{$\qchk{\CTX}{\tinsp(\theta)}{\sigma}$} 
  \DisplayProof
\end{tabular}
\end{center}
}}
\caption{Typing rules for \JLs{} terms}\label{fig:typerules_simple}
\begin{tabular}{c}
\rVS
\AxiomC{$((\lambda a.s) \; t) \red s\subst{a}{t}$}
\DisplayProof
\qquad
\AxiomC{$\tlet(\bang s,u.t) \red s\subst{u}{t}$}
\DisplayProof
\qquad
\AxiomC{$\tinsp(\theta) \red q\theta$}
\DisplayProof
\qquad
\AxiomC{$s \red t$}
\UnaryInfC{$\cE[s] \red \cE[t]$}
\DisplayProof
\end{tabular}
\caption{Reduction rules for \JLs{} terms}\label{fig:termredrules_simple}

\end{figure*}

The typing rules and reduction rules for \JLs{} are given in
Figures~\ref{fig:typerules_simple} and~\ref{fig:termredrules_simple}. They are similar to their
counterparts in \JLr, but greatly simplified: in particular, trail
inspection is allowed to reduce by nondeterministically picking any trail.

Unsurprisingly, erasure preserves well-typedness and reduction:
\begin{lemma}\label{lem:erasetchk}
If $\CTX \vdash_\JLr M : A | s$, then $\CTX \vdash_\JLs |M| : |A|$.
\end{lemma}
\begin{lemma}\label{lem:erasered}
$M \red_\JLr N \Longrightarrow |M| \red_\JLs |N|$
\end{lemma}

By the combination of Lemma~\ref{lem:erasetchk} and Lemma~\ref{lem:erasered}, we know that if \JLs{} is strongly normalizing, then \JLr{} must also, \emph{a fortiori}, be strongly normalizing. We will thus proceed to prove SN in the simpler system, and extend it to \JLr as a corollary.

\subsection{Summary of the proof}
We now give the definition of candidates of reducibility \emph{\`a la
  Girard}. 

\begin{definition}[neutral term]
A term is neutral if it is not of the form $\lambda a^A.s$ or $! s$.
\end{definition}

\begin{definition}[candidates of reducibility]\label{defn:cr}
A set $\cC$ of terms is a \emph{candidate of reducibility} iff it satisfies Girard's CR conditions:
\begin{description}
\item[CR1] $\cC \subseteq \SN$
\item[CR2] $s \in \cC \land s \red t \Longrightarrow t \in \cC$
\item[CR3] $s \in \NT \land (\forall t. s \red t \Longrightarrow t \in \cC) \Longrightarrow s \in \cC$.
\end{description}
The set of candidates of reducibility will be denoted $\CR$.
\end{definition}

We identify certain subsets of candidates that are stable wrt. validity substitution: 
\begin{definition}[substitutive sub-candidate]
For all candidates $\cC$, validity variables $u$, and sets of terms $\cD$, we define its substitutive subset $\subCR{\cC}{u}{\cD}$ as $\subCR{\cC}{u}{\cD} \triangleq \{ s\in \cC : \forall t \in \cD, s\subst{u}{t} \in \cC\}$.
\end{definition}


Notice that sub-candidates are not in $\CR$; they do, however, satisfy CR1 and CR2, which is enough for us to use them in the definition of \emph{reducible sets}.

\begin{definition}[reducible set]
For all types $\tau$, the set $\Red{\tau}$ of reducible terms of type $\tau$ is defined by recursion on $\tau$ as follows:
\begin{itemize}
\item $\Red{P} = \SN$
\item $\Red{\tau \imp \sigma} = \{ s : \forall t \in \Red{\tau}, (s~t) \in \Red{\sigma} \}$
\item $\Red{\Box \tau} = \{ s : \forall u, \forall \cC \in \CR, \forall t \in \subCR{\cC}{u}{\Red{\tau}}, \tlet(u := s, t) \in \cC \}$
\end{itemize}
\end{definition}

In particular, $s \in \Red{\Box \tau}$ if, and only if, for all
reducibility candidates $\cC$ and audited variables $u$, if we take a
term $t \in \subCR{\cC}{u}{\Red{\tau}}$, then $\tlet(u := s, t) \in \cC$.
We are allowed to use $\Red{\tau}$ in $\subCR{\cC}{u}{\Red{\tau}}$ because
$\tau$ is a subexpression of $\Box \tau$.
thus we need to quantify over all candidates $\cC$.

\begin{lemma}\label{lem:redCR}
For each type $\tau$, $\Red{\tau} \in \CR$.
\end{lemma}

\begin{theorem}\label{thm:snmain}
Let $\Delta,\Gamma,\vect{\eta}$ s.t. $\dom(\vect{\eta}) = \dom(\Delta) \cup \dom(\Gamma)$, for all $u \in \dom(\Delta), \vect{\eta}(u) \in \Red{\Delta(u)}$, and for all $a \in \dom(\Gamma), \vect{\eta}(a) \in \Red{\Gamma(a)}$. Then, $\Delta; \Gamma \vdash s : \tau$ implies $s \vect{\eta} \in \Red{\tau}$.
\end{theorem}
\begin{proof}
We use the standard technique for System F~\cite{proofstypes}, with additional subproofs for $s \in \Red{\tau} \implies \bang s \in \Red{\Box \tau}$ and $(\forall \psi \in \dom(\theta), \theta(\psi) \in \Red{\cT^\sigma(\psi)}) \implies \tinsp(\theta) \in \Red{\sigma}$. SN follows as a corollary when $\vect{\eta}$ is the identity substitution.
\end{proof}

\subsection{Formalization}
We formalized \JLr, together with our proof of strong normalization,
in Nominal Isabelle. Nominal Isabelle mechanizes the management of
variable binding, relieving the user from the burden of defining
binding infrastructure (e.g. de~Bruijn indices, lifting operations,
etc.). On the other hand, many of the definitions used in a nominal
formalization must be proved to be ``well-behaved'', beyond what would
usually be made explicit in a pencil-and-paper proof.  The main
well-behavedness property required in Nominal Isabelle is
\emph{equivariance}, stating that a function $f$ or a set $\cS$ is
stable under finite permutations of names $\pi$:
\[
\forall x.f(\pi \cdot x) = \pi \cdot f(x)  \qquad  
\forall x.x \in \cS \iff \pi \cdot x \in \cS
\]
Most of the definitions used in the proof are equivariant, including typing and reduction rules, the set $\SN$ of strongly normalizing terms, the set $\CR$ of reducibility candidates, reducibility sets $\Red{A}$ for all types $A$, and the operator $\subCR{\cC}{u}{\cD}$.

We do not, however, prove equivariance for individual reducibility candidates $\cC \in \CR$: on the contrary, reducibility candidates do not need to be equivariant. To prove it, one can take a closure operator $[\cdot]$ mapping a set of strongly normalizing terms to the smallest reducibility candidate containing it (for its existence and definition see e.g.~\cite{Riba2007}): it is easy to see that, for all variables $a, b$, $\bang a \in [\{\bang b \}]$ if and only if $a = b$.

Our impredicative definition of $\Red{\Box \tau}$, which quantifies
over arbitrary candidates, is handled gracefully by Nominal
Isabelle. Lindley and Stark~\cite{Lindley2005} show a technique
($\top\top$-lifting) that can be adapted to provide a predicative
definition of $\Red{\Box A}$. As it happens, the lower logical
complexity of $\top\top$-lifting relies on the additional definition
of \emph{continuations}, which would require some more effort to be
formalized.  On the other hand, our approach seems likely to extend
to handle structural recursion (System T),
following Aschieri and Zorzi~\cite{aschieri13tlca}, or impredicative
polymorphism (System F), following Girard et al.~\cite{proofstypes}.

\section{Related work}\label{sec:related}
The first Justification Logic-style system, known as the Logic of
Proofs, was introduced by Artemov~\cite{Artemov2001,Artemov2008}.
Most work on justification logic systems (as for modal logic) is
presented via Hilbert-style axiom systems extending classical
propositional logic. Pfenning and Davies' judgemental reconstruction
of modal logic~\cite{Pfenning2001} provides a natural deduction-style proof system for
(intuitionistic) modal logic with necessity and possibility
modalities.  Artemov and Bonelli~\cite{Artemov2007} introduced a
system $\lambda^\mathbf{I}$ for justification logic, extending 
Pfenning and Davies' treatment of necessity ($\Box A$).  They introduced
equality proofs (trails) to recover subject reduction and proved
strong normalization for $\lambda^\mathbf{I}$. Bavera and Bonelli later introduced (outer) trail inspection as part of an extended calculus called \JLh~\cite{Bavera2015} from which we took inspiration.

Our system \JLr seems (at least to us) an improvement over \JLh: it is
simpler, and avoids the inconsistency arising from viewing trails as
equivalence proofs.  Nevertheless, the technique we used to prove
strong normalization in \JLr seems to suffice for \JLh as well. The
only complication concerns the definition of reduction: while in \JLr
reduction acts non-locally by updating the trail in the nearest
enclosing bang, in \JLh reduction produces an intermediate term
annotated with a new local trail, and the trail in the enclosing bang
is updated after a sequence of \emph{permutation reductions}:
\[
\bang_q \cF[M] \red \bang_q \cF[q' \derives M'] \stackrel{*}{\red} \bang_{q''} \cF[M']
\]
It is thus necessary, though not overly complicated, to redefine
reduction as a single-step operation including permutation reductions,
and prove its equivalence to the original definition. This allows us
to remove intermediate terms $q \derives M$ from the calculus
altogether.  In general, the (efficient) reduction theory of systems
such as \JLh and \JLr deserves further study.

Our work is also partly motivated by previous work on provenance and
tracing for functional languages.  Perera et al.~\cite{Perera2012}
presented a pure, ML-like core language in which program execution
yields both a value and a \emph{trace}, corresponding roughly to the
large-step operational derivation leading to the result.  They gave
trace slicing algorithms and techniques for extracting program slices
and differential slices from traces; in subsequent work Acar et
al.~\cite{Acar2013} explored security implications such as the
disclosure and obfuscation properties of traces~\cite{Cheney2011}.
Indeed, trail inspection can be considered as a generic
mechanism for defining \emph{provenance views} as considered by these
papers.  Although Bavera and Bonelli~\cite{Bavera2015} motivated \JLh
as a basis for history-based access control (following Abadi and
Fournet~\cite{Abadi2003}), we are not aware of comparable work on
provenance security based on justification logic. Our ongoing 
investigation suggests that \JLh-style trails contain enough 
information to extract many forms of provenance; however, to perform 
this extraction by means of trail inspection, we would usually need to
reverse beta-reductions, and in particular to undo the substitution
in beta-reduced terms. Since inspections treat beta and beta-box 
trails as black boxes, this cannot be achieved in the current version
of the calculus. An extension of the language providing transparent 
beta trails and operations to undo substitutions is the subject of 
our current study.

Audit has been considered by a number of security researchers
recently, for example Amir-Mohamedian et
al.~\cite{amir-mohamedian2016} consider correctness for audit logging,
but auditing is again an extralinguistic operation (implemented using
aspect-oriented programming).  Vaughan et al.~\cite{vaughan08csf}
introduced new formalisms for evidence-based audit.  In Aura, an
implementation of this approach~\cite{jia08icfp}, dependently-typed
programs execute in the presence of a policy specified in
authorization logic~\cite{abadi93toplas}, and whenever a restricted
resource is requested, a proof that access is authorized is
constructed at run time and stored in an audit log for later
inspection. The relationship between this approach and ours, and more
generally between justification logic and authorization logic, remains
to be investigated.

\section{Conclusions}
The motivation for this work is the need to provide better foundations
for audited computation, as advocated in recent work on provenance and
security and on type-theoretic forms of justification logic such as \JLh. However,
as we have shown, \JLh is at the same time overly restrictive (in
requiring affine variable and trail variable usage, i.e. forbidding
copying of ordinary data) and overly permissive: despite these
restrictions aimed at keeping the equational theory consistent, one
can still prove any two compatibly-typed proof codes equivalent, using
symmetry and the nondeterministic equational law for trail inspection.  

We presented a new calculus \JLr based on justification logic that
includes audit operations such as trail inspection, but has fewer
limitations and is simpler than \JLh.  We show that \JLh avoids this
problem and has a consistent reduction theory.  We also prove strong
normalization for \JLr via a simplified system \JLs, and we have
mechanically checked the proof.  This approach also seems sufficient
to prove SN for \JLh, though we have not formalized this result.

In future work, we intend to consider larger-scale programming
languages based on the ideas of \JLr, investigate connections to
provenance tracking and slicing techniques, and clarify the security
guarantees offered by justification logic-based audit.





\bibliography{paper}


\newpage

\appendix

\section{The type system of \JLh}\label{app:jlh}
In Figures~\ref{fig:codeJLh}, \ref{fig:trailJLh},
and~\ref{fig:termJLh}, we provide the typing rules of Bavera and
Bonelli's \JLh for reference, using syntax similar to that used in
this paper.

\begin{figure*}[p]
 \begin{tabular}{c}
 \rVS
 \AxiomC{$a:A \in \Gamma$} \RightLabel{\textsc{Var}}
 \UnaryInfC{$\qchkh{\CTX;\Sigma}{a}{A}$}
 \DisplayProof
 \quad
 \AxiomC{$u:A[\Sigma] \in \Delta$} 
 \AxiomC{$\Sigma\sigma \subseteq \Sigma'$} \RightLabel{\textsc{mVar}}
 \BinaryInfC{$\qchkh{\Delta;\Gamma;\Sigma'}{\pair{u,\sigma}}{A}$}
 \DisplayProof
 \\\rVS
 \AxiomC{$\qchkh{\Delta;\Gamma,a:A;\Sigma}{s}{B}$} \RightLabel{$\imp \Intro$}
 \UnaryInfC{$\qchkh{\CTX}{\lambda a^A.s}{A \imp B}$}
 \DisplayProof
 \quad
 \AxiomC{$\qchkh{\Delta;\Gamma_1;\Sigma_1}{s}{A \imp B}$} 
 \AxiomC{$\qchkh{\Delta;\Gamma_2;\Sigma_2}{t}{A}$}
  \RightLabel{$\imp \Elim$}
 \BinaryInfC{$\qchkh{\Delta;\Gamma_1,\Gamma_2;\Sigma_1,\Sigma_2}{s\;t}{B}$}
 \DisplayProof
 \\\rVS
 \AxiomC{$\qchkh{\Delta;\cdot;\Sigma}{s}{A}$}
 \AxiomC{$\qchkh{\Delta;\cdot;\Sigma}{q}{s =_A t}$} 
   \RightLabel{$\Box \Intro$}
 \BinaryInfC{$\qchkh{\Delta;\Gamma;\Sigma'}{\Sigma.t}{\AU{\Sigma.t}A}$} 
 \DisplayProof
 \quad
 \AxiomC{$\qchkh{\Delta;\Gamma_1;\Sigma_1}{s}{\AU{\Sigma.r} A}$} 
 \AxiomC{$\qchkh{\Delta,u:A[\Sigma];\Gamma_2;\Sigma_2}{t}{C}$} 
   \RightLabel{$\Box \Elim$}
 \BinaryInfC{$\qchkh{\Delta;\Gamma_1,\Gamma_2;\Sigma_1,\Sigma_2}{\tlet(u^{A[\Sigma]} := s, t)}{C\subst{u}{\Sigma.r}}$}
 \DisplayProof
 \\
 \rVS
 \AxiomC{$\alpha : \Eq(A) \in \Sigma$}
 \AxiomC{$\qchkh{\Delta;\cdot;\cdot}{\theta}{\cT^B}$} 
   \RightLabel{TI}
 \BinaryInfC{$\qchkh{\CTX;\Sigma}{\alpha\theta}{B}$} 
 \DisplayProof
 \quad
 \AxiomC{$\qchkh{\Delta;\Gamma;\Sigma}{s}{A}$}
 \AxiomC{$\qchkh{\Delta;\Gamma;\Sigma}{q}{s =_A t}$} 
   \RightLabel{Eq}
 \BinaryInfC{$\qchkh{\CTX;\Sigma}{t}{A}$} 
 \DisplayProof
\end{tabular}
\caption{Typing rules for \JLh proof codes}\label{fig:codeJLh}
 \begin{tabular}{c}
 \rVS
 \AxiomC{$\qchkh{\CTX;\Sigma}{s}{A}$} \RightLabel{\textsc{EqRefl}}
 \UnaryInfC{$\qchkh{\CTX;\Sigma}{\qrefl(s)}{s =_A s}$}
 \DisplayProof
 \quad
 \AxiomC{$\qchkh{\Delta;\Gamma_1,a:A;\Sigma_1}{s}{A \imp B}$}
 \AxiomC{$\qchkh{\Delta;\Gamma_2;\Sigma_2}{t}{A}$}
   \RightLabel{\textsc{Eq$\qba$}}
 \BinaryInfC{$\qchkh{\Delta;\Gamma_1,\Gamma_2;\Sigma_1,\Sigma_2}{\qba(a^A.s,t)}{s\subst{a}{t} =_B (\lambda a^A.s) \; t}$}
 \DisplayProof
 \\
 \rVS
 \AxiomC{$\Delta;\cdot;\Sigma_1 \vdash A | r$} 
 \AxiomC{$\Delta,u:A[\Sigma_1];\Gamma_2;\Sigma_2 \vdash C | t$}
 \AxiomC{$\Delta;\cdot;\Sigma_1 \vdash q : r =_A s$} 
 \AxiomC{$\Gamma_2 \subseteq \Gamma_3 \quad \Sigma_2 \subseteq \Sigma_3$} 
   \RightLabel{\textsc{Eq$\qbb$}}
 \QuaternaryInfC{$\qchkh{\Delta;\Gamma_3;\Sigma_3}
              {\qbb(\Sigma_1.s, u^{A[\Sigma_1]}.t)}
              {t\subst{u}{\Sigma_1.s} =_{C\subst{u}{\Sigma_1.s}}
              \tlet(u^{A[\Sigma_1]} := s,t)}$}
 \DisplayProof
 \\
 \rVS
 \AxiomC{$\Delta;\cdot;\Sigma_1 \vdash q : s =_A t $}
 \AxiomC{$\Delta;\cdot;\cdot \vdash \cT^B | \theta$}
 \AxiomC{$\alpha : \Eq(A) \in \Sigma_2$}
  \RightLabel{$\textsc{EqTI}$}
 \TrinaryInfC{$\qchkh{\Delta;\Gamma;\Sigma_2}{\qti(\vartheta,\alpha)}
  {q\theta =_B \alpha\theta}$}
 \DisplayProof
 \\
 \rVS
 \AxiomC{$\qchkh{\CTX;\Sigma}{q}{s =_A t}$} \RightLabel{\textsc{EqSym}}
 \UnaryInfC{$\qchkh{\CTX;\Sigma}{\qsym(q)}{t =_A s}$}
 \DisplayProof
 \quad
 \AxiomC{$\qchkh{\CTX;\Sigma}{q_1}{r =_A s}$} 
 \AxiomC{$\qchkh{\CTX;\Sigma}{q_2}{s =_A t}$} 
   \RightLabel{\textsc{EqTrans}}
 \BinaryInfC{$\qchkh{\CTX;\Sigma}{\qtrans(q_1,q_2)}{r =_A t}$}
 \DisplayProof
 \\
 \rVS
 \AxiomC{$\qchkh{\Delta;\Gamma,a:A;\Sigma}{q}{s =_B t}$}
   \RightLabel{\textsc{EqAbs}}
 \UnaryInfC{$\qchkh{\CTX;\Sigma}{\qlam(a^A.q)}{\lambda a^A.s =_{A \imp B} \lambda a^A.t}$} 
 \DisplayProof
 \\
\rVS
 \AxiomC{$\qchkh{\Delta;\Gamma_1;\Sigma_1}{q_1}{s_1 =_{A \imp B} t_1}$} 
 \AxiomC{$\qchkh{\Delta;\Gamma_2;\Sigma_2}{q_2}{s_2 =_A t_2}$}
   \RightLabel{\textsc{EqApp}}
 \BinaryInfC{$\qchkh{\Delta;\Gamma_1,\Gamma_2;\Sigma_1,\Sigma_2}{\qapp(q_1,q_2)}{s_1 \; s_2 =_B t_1 \; t_2}$} 
 \DisplayProof
 \\
 \rVS
 \AxiomC{$\qchkh{\Delta;\Gamma_1;\Sigma_1}{q_1}{s_1 =_{\AU{\Sigma_r} A} t_1}$}
 \AxiomC{$\qchkh{\Delta, u:A[\Sigma];\Gamma_2;\Sigma_2}{q_2}{s_2 =_C t_2}$}
   \RightLabel{\textsc{EqLet}}
 \BinaryInfC{$\qchkh{\Delta;\Gamma_1,\Gamma_2;\Sigma_1,\Sigma_2}{\qlet(q_1,u^{A[\Sigma]}.q_2)}{\tlet(u^{A[\Sigma]} := s_1,s_2) 
    =_{C\subst{u}{\Sigma.r}}
    \tlet(u^{A[\Sigma]} := t_1,t_2)}$} 
 \DisplayProof
 \\
 \rVS
 \AxiomC{$\alpha : \Eq(A) \in \Sigma$}
 \AxiomC{$\qchkh{\Delta;\cdot;\cdot}{\vect{q}}{\theta =_{\cT^B} \theta'}$}
   \RightLabel{\textsc{EqTrpl}}
 \BinaryInfC{$\qchkh{\CTX;\Sigma}{\qtr(\alpha,\vect{q})}{\alpha\theta =_B \alpha\theta'}$} 
 \DisplayProof
\end{tabular}
\caption{Typing rules for \JLh trails}\label{fig:trailJLh}
 \begin{tabular}{c}
 \rVS
 \AxiomC{$a:A \in \Gamma$} \RightLabel{\textsc{TVar}}
 \UnaryInfC{$\tchkh{\CTX;\Sigma}{a}{A}{a}$}
 \DisplayProof
 \quad
 \AxiomC{$u:A[\Sigma] \in \Delta$} 
 \AxiomC{$\Sigma\sigma \subseteq \Sigma'$} \RightLabel{\textsc{TmVar}}
 \BinaryInfC{$\tchkh{\Delta;\cdot;\Sigma'}{\pair{u,\sigma}}{A}{\pair{u,
   \sigma}}$}
 \DisplayProof
 \\
\rVS
 \AxiomC{$\tchkh{\Delta;\Gamma,a:A;\Sigma}{M}{B}{s}$} \RightLabel{\textsc{TAbs}}
 \UnaryInfC{$\tchkh{\CTX;\Sigma}{\lambda a^A.M}{A \imp B}{\lambda a^A.s}$}
 \DisplayProof
 \quad
 \AxiomC{$\tchkh{\Delta;\Gamma_1;\Sigma_1}{M}{A \imp B}{s}$}
 \AxiomC{$\tchkh{\Delta;\Gamma_2;\Sigma_2}{N}{A}{t}$}
  \RightLabel{\textsc{TApp}}
 \BinaryInfC{$\tchkh{\Delta;\Gamma_1,\Gamma_2;\Sigma_1,\Sigma_2}{M\;N}{B}{s\;t}$}
 \DisplayProof
 \\ \rVS
 \AxiomC{$\tchkh{\Delta;\cdot;\Sigma}{M}{A}{s}$}
 \AxiomC{$\qchkh{\Delta;\cdot;\Sigma}{q}{s = t}$} \RightLabel{\textsc{TBox}}
 \BinaryInfC{$\tchkh{\Delta;\Gamma;\Sigma'}{!^\Sigma_q M}{\AU{\Sigma.t}A}
   {\Sigma.t}$}
 \DisplayProof
 \\
 \rVS
 \AxiomC{$\tchkh{\Delta;\Gamma_1;\Sigma_1}{M}{\AU{\Sigma.r}A}{s}$} 
 \AxiomC{$\tchkh{\Delta,u:A[\Sigma];\Gamma_2;\Sigma_2}{N}{C}{t}$} 
   \RightLabel{\textsc{TLet}}
 \BinaryInfC{$\tchkh{\Delta;\Gamma_1,\Gamma_2;\Sigma_1,\Sigma_2}{\tlet(u^{A[\Sigma]} := 
   M,N)}{C\subst{u}{\Sigma.r}}{\tlet(u^{A[\Sigma]} := s,t)}$}
 \DisplayProof
\\ \rVS
 \AxiomC{$\alpha : \Eq(A) \in \Sigma$} 
 \AxiomC{$\tchkh{\Delta;\cdot;\cdot}{\vartheta}{\cT^B}{\theta}$} 
   \RightLabel{\textsc{TTI}}
 \BinaryInfC{$\tchkh{\CTX;\Sigma}{\alpha\vartheta}{B}{\alpha\theta}$} 
 \DisplayProof
\quad
 \AxiomC{$\tchkh{\Delta;\Gamma;\Sigma}{M}{A}{s}$} 
 \AxiomC{$\qchkh{\Delta;\Gamma;\Sigma}{q}{s =_A t}$} 
   \RightLabel{\textsc{TEq}}
 \BinaryInfC{$\tchkh{\CTX;\Sigma}{q \derives M}{A}{t}$}
 \DisplayProof
\end{tabular}
\caption{Typing rules for \JLh terms}\label{fig:termJLh}
\end{figure*}

\subsection{Inconsistency of the equational theory of \JLh}
In Section~\ref{sec:intro}, we provided an informal derivation tree
showing that the equational theory of \JLh allows to equate any two
proof codes of the same type. We give below a more accurate (but isomorphic)
derivation tree. Then, if we choose a $\theta$ such that
$\theta(\qrefl) = s$ and $\theta(\qba) = t$, this effectively proves
$s = t$, where $s$ and $t$ can be chosen freely:

\begin{center}
\small
	\begin{prooftree}
      \AxiomC{$\qchkh{\cdot;\cdot;\cdot}{\qrefl(s')}{s' =_B s'}$} \noLine
      \UnaryInfC{$\qchkh{\cdot;\cdot;\cdot}{\theta}{\cT^A}$}
\RightLabel{\textsc{EqTI}}
	\UnaryInfC{$\qchkh{\cdot;\cdot;\alpha:\Eq(A)}{\qti(\theta,\alpha)}{\qrefl(s')\theta =_A \alpha\theta}$}
	
      \AxiomC{$\qchkh{\cdot;\cdot;\cdot}{\qba(b^{B'}.t',t'')}{t'\subst{b}{t''} =_{B''} (\lambda b^{B'}.t')~t''}$}
      \noLine
      \UnaryInfC{$\qchkh{\cdot;\cdot;\cdot}{\theta}{\cT^A}$}
\RightLabel{\textsc{EqTI}}
	\UnaryInfC{$\qchkh{\cdot;\cdot;\alpha:\Eq(A)}{\qti(\theta,\alpha)}{\qba(b^{B'}.t',t'')\theta =_A \alpha\theta}$}
\RightLabel{\textsc{EqSym}}
	\UnaryInfC{$\qchkh{\cdot;\cdot;\alpha:\Eq(A)}{\qsym(\qti(\theta,\alpha))}{\alpha\theta =_A \qba(b^{B'}.t',t'')\theta}$}
	
\RightLabel{\textsc{EqTrans}}
	\BinaryInfC{$\qchkh{\cdot;\cdot;\alpha:\Eq(A)}{\qtrans(\qti(\theta,\alpha),\qsym(\qti(\theta,\alpha)))}{\qrefl(s')\theta =_A \qba(b^{B'}.t',t'')\theta}$}
	\end{prooftree}
\end{center}
This derivation contains a lot of extraneous details, so we summarize
the main features which contribute to the final result.  In this
derivation, we use transitivity on two subderivations.  The right-hand
subderivation shows that $s = \qrefl(s')\theta$ is equal to
$\alpha \theta$ , using the rule \textsc{EqTI}.  Here, $s'$ is an
arbitrary well-formed term of some type $B$.  
The \textsc{EqTI}
rule allows us to conclude that any code of the form $q\theta$ is
equal to $\alpha\theta$, and here we have used $q = \qrefl(s')$.  

In the second subderivation we use rule \textsc{EqTI} again to
conclude that $t = \qba(b^{B'}.t',t'')\theta$ is equivalent to $\alpha\theta$.
Here, terms $t'$ ad $t''$ are arbitrary well-formed terms of
appropriate types to form a redex $\lambda b. t')~t''$.  We then use
symmetry to show that $\alpha\theta$ is equivalent to
$\qba(b^{B'}.t',t'')\theta$.  The desired conclusion follows.  Since we
can easily choose $\theta$ so that it returns one proof code $s$ when
handling a reflexivity trail and a different proof code $t$ when
handling a $\beta$-redex trail, the above derivation can be modified
to prove that any two proof codes (of the same type $A$) are
equivalent.  

Finally, in applying the transitivity rule \textsc{EqTrans}, it is
important that this rule does \emph{not} require affine use of the
trail variable $\alpha$, so that it can be used in both
subderivations.  Changing this might fix the above problem, but would
also have unknown consequences on the rest of the system.  We have
instead focused on an alternative system that does not require trail
variables or affineness, by removing the symmetry trail.

\section{Full definitions}


\subsection{Substitution on codes}
\label{app:codesubst}
For $\gamma = \subst{a}{t}$ and $\delta = \subst{u}{t}$, we define substitution of simple and audited variables on a code $r$ (notation: $r\gamma$ and $r\delta$) as follows:
{\small{
\begin{align*}
b \gamma & \triangleq \gamma(b) \\
v \gamma & \triangleq v \\
(\lambda b^A.s)\gamma & \triangleq \lambda b^A.s\gamma & \mbox{($b \# a,t$)} \\
(s \; s')\gamma & \triangleq (s\gamma) \; (s'\gamma) \\
(\bang s)\gamma & \triangleq \bang s \\
\tlet(v^A := s, s')\gamma & \triangleq \tlet(v^A := s\gamma,s'\gamma) & \mbox{($v \# t$)} \\
\tinsp(\theta)\gamma & \triangleq \tinsp(\theta\gamma) \\
\\
b \delta & \triangleq b  
\\
v \delta & \triangleq \delta(v) 
\\
(\lambda b^A.s)\delta & \triangleq \lambda b^A.s\delta & \mbox{($b \# t$)}
\\
(s \; s')\delta & \triangleq (s\delta) \; (s'\delta)
\\
(\bang s)\delta & \triangleq \bang (s\delta)
\\
\tlet(v^A := s, s')\delta & \triangleq \tlet(v^A := s\delta,s'\delta) & \mbox{($v \# u,t$)}
\\
\tinsp(\theta)\delta & \triangleq \tinsp(\theta\delta)
\end{align*}
}}
where $\theta\gamma$ and $\theta\delta$ are defined pointwise and the notation $\gamma(b)$ is defined as $t$ if $a=b$, and as $b$ otherwise ($\delta(v)$ is defined similarly).

\subsection{Source and target of a trail}
\label{app:srctgt}

The source and target of a trail $q$ (notation: $\src(q)$ and $\tgt(q)$) are defined as follows:
{\small{
\begin{align*}
\src(\qrefl(s)) & \triangleq s 
\\
\src(\qtrans(q_1,q_2)) & \triangleq \src(q_1)
\\ 
\src(\qba(a^A.s_1,s_2)) & \triangleq (\lambda a^A.s_1)~s_2
\\
\src(\qbb(s_1,v^A.s_2)) & \triangleq \tlet(v^A := \bang s_1, s_2) 
\\
\src(\qoti(q',\theta)) & \triangleq \tinsp(\theta)
\\
\src(\qlam(a^A.q')) & \triangleq \lambda a^A.\src(q')
\\
\src(\qapp(q', q'')) & \triangleq (\src(q')~\src(q''))
\\
\src(\qlet(q',v^A.q'')) & \triangleq \tlet(v^A := \src(q'), \src(q''))
\\
\src(\qotr(\zeta)) & \triangleq \tinsp(\src(\zeta))
\\
\\
\tgt(\qrefl(s)) & \triangleq s
\\
\tgt(\qtrans(q_1,q_2)) & \triangleq \tgt(q_2)
\\
\tgt(\qba(a^A.s_1,s_2)) & \triangleq s_1\subst{a}{s_2}
\\
\tgt(\qbb(s_1,v^A.s_2)) & \triangleq s_2\subst{v}{s_1}
\\
\tgt(\qoti(q',\theta)) & \triangleq q'\theta
\\
\tgt(\qlam(a^A.q')) & \triangleq \lambda a^A.\tgt(q') 
\\
\tgt(\qapp(q', q'')) & \triangleq (\tgt(q')~\tgt(q'')) 
\\
\tgt(\qlet(q',v^A.q'')) & \triangleq \tlet(v^A := \tgt(q'), \tgt(q''))
\\
\tgt(\qotr(\zeta)) & \triangleq \tinsp(\tgt(\zeta))
\end{align*}
}}
where $\src(\zeta)$ and $\tgt(\zeta)$ are defined pointwise.

\subsection{Audited variable substitution}
\label{app:avarsubst}

Given $\delta = \subst{u}{(N,q,t)}$, fix $\delta' = \subst{u}{t}$. Then, the operations $M \times \delta$ and $M \ltimes \delta$ are defined as follows:
{\small{
\begin{align*}
a \times \delta &\triangleq \qrefl(a)
\\
v \times\delta &\triangleq \begin{cases}
  q & (u = v) \\
  \qrefl(v) & (u \neq v)
  \end{cases} 
\\
(\lambda a^A.R) \times\delta &\triangleq \qlam(a^A.R \times\delta) & (\heartsuit)
\\
(R \; S) \times \delta &\triangleq \qapp(R \times\delta, S \times\delta)
\\
(\bang_{q'} R) \times \delta &\triangleq \qrefl(\bang (\src(q')\delta'))
\\
\tlet(v^A := R, S) \times\delta &\triangleq \qlet(R \times\delta,v^A.S \times\delta) & (\spadesuit)
\\
\tinsp(\vartheta) \times \delta &\triangleq \qotr(\vartheta\times\delta)
\\
\\
a \ltimes \delta & \triangleq a
\\
v \ltimes\delta & \triangleq \begin{cases}
  N & (u = v) \\
  v & (u \neq v)
  \end{cases} 
\\
(\lambda a^A.R) \ltimes\delta & \triangleq \lambda a^A.R \ltimes\delta & (\heartsuit)
\\
(R \; S) \ltimes \delta & \triangleq (R \ltimes\delta) \; (S \ltimes\delta)
\\
(\bang_{q'} R) \ltimes \delta & \triangleq \bang_{\qtrans(q'\delta',R \times\delta)} (R \ltimes \delta)
\\
\tlet(v^A := R, S) \ltimes\delta & \triangleq \tlet(v^A := R \ltimes\delta,u.S \ltimes\delta) & (\spadesuit)
\\
\tinsp(\vartheta) \ltimes \delta & \triangleq \tinsp(\vartheta \ltimes\delta)
\end{align*}
\[
\begin{array}{llll}
(\heartsuit) & a \# N,q,t &
(\spadesuit) & v \# u,N,q,t
\end{array}
\]
}}
where $\vartheta\times\delta$ and $\vartheta\ltimes\delta$ are defined pointwise on $\vartheta$.

\subsection{Codes}
\label{app:codes}

The code of a term $M$ (notation: $\code(M)$) is defined as follows:
{\small{
\begin{align*}
\code(a) & \triangleq a 
\\
\code(u) & \triangleq u
\\
\code(\lambda a^A.N) & \triangleq \lambda a^a.\code(N)
\\
\code(N~R) & \triangleq (\code(N)~\code(R))
\\
\code(\bang_q N) & \triangleq \bang\src(q)
\\
\code(\tlet(u^A := N, R)) & \triangleq \tlet(u^A := \code(N), \code(R))
\\
\code(\tinsp(\vartheta)) & \triangleq \tinsp(\code(\vartheta))
\end{align*}
}}
where $\code(\vartheta)$ is defined pointwise.

\section{Proofs}

\subsection{Strong normalization}
The proofs of strong normalization and related lemmata were only
sketched because of space constraints. This appendix provides more
detailed proofs and explanations, which summarize the mechanized
formalization in Nominal Isabelle. For convenience, we explicitly
state the well-founded induction principle we use in some of the
following proofs:

\begin{theorem}[induction principle for well-founded relations]\label{thm:wfind}
Let $R \subseteq X \times X$ a well-founded relation on some set $X$, and $P \subseteq X$ a predicate on the same set. Suppose that if $x \in X$ and for all $y \in X$ such that $y \mathrel{R} x$ we have $P~y$, then $P~x$. Then, $P~z$ holds for all $z \in X$.
\end{theorem}

The principle can also be applied to tuples of strongly normalizing terms by taking the lexicographic extension of the reduction relation, which is also well-founded.

Some standard results hold in \JLs. The following substitution lemma can be proved by a routine induction.
\begin{lemma}[substitution lemma]\label{lem:subst}\mbox{}
\begin{enumerate}
\item $a \# r \Longrightarrow s\subst{a}{t}\subst{u}{r} = s\subst{u}{r}\subst{a}{t\subst{u}{r}}$
\item $u \# v,r \Longrightarrow s\subst{u}{t}\subst{v}{r} = s\subst{v}{r}\subst{u}{t\subst{v}{r}}$
\end{enumerate}
\end{lemma}

Consequently, we can prove the substitutivity of reduction.
\begin{lemma}[substitutivity of reduction]\label{lem:redsubst}
$s \red t \Longrightarrow s\subst{u}{r} \red t\subst{u}{r}$
\end{lemma}
\begin{proof}
By induction on the derivation of reduction, using Lemma~\ref{lem:subst}.
\end{proof}

We mentioned that our proof of strong normalization uses candidates of reducibility (Definition~\ref{defn:cr}): however, we found it convenient to also use a more general notion that we call
\emph{pre-candidate of reducibility} after Petit~\cite{Petit2009}.

\begin{definition}[pre-candidates of reducibility]
A set $\cC$ of terms is a \emph{pre-candidate of reducibility} iff it satisfies Girard's CR1 and CR2 conditions:
\begin{description}
\item[CR1] $\cC \subseteq \SN$
\item[CR2] $s \in \cC \land s \red t \Longrightarrow t \in \cC$
\end{description}
The set of pre-candidates of reducibility will be denoted $\PCR$.
\end{definition}

Substitutive sub-candidates of a candidate are pre-candidates.
\begin{lemma}\label{lem:subCRisPCR}
$\subCR{\cC}{u}{\cD} \in \PCR$
\end{lemma}
\begin{proof}
Since $\subCR{\cC}{u}{\cD} \subseteq \cC$, CR1 follows immediately. To prove CR2, we assume $s \in \subCR{\cC}{u}{\cD}$ and $s \red s'$. By Lemma~\ref{lem:redsubst}, we prove that for all $t \in \cD$, $s\subst{u}{t} \red s'\subst{u}{t}$. From CR2 on the candidate $\cC$, we know $s' \in \cC$, thus $s' \in \subCR{\cC}{u}{\cD}$, which implies the thesis.
\end{proof}

\begin{remark}
We easily see that CR3 cannot hold for $\subCR{\cC}{u}{\cD}$, meaning it is a pre-candidate but not a candidate, even if $\cD$ were assumed to be a candidate. Otherwise, we would know $u \in \subCR{\cC}{u}{\cD}$, and consequently for all $N \in \cD$ we would have $u\subst{u}{N} = N \in \cC$. Hence, all candidates would need to be equal, but we know that there exist at least two different candidates: $\SN$ and the set of all variables.
\end{remark}

Some other immediate results about reducibility candidates, whose proof is standard, follow.
\begin{lemma}\mbox{}
\begin{enumerate}
\item $\SN \in \CR$
\item for all $\cC \in \CR$ and $s \in \NT$ in normal form, $s \in \cC$
\item in particular, for all reducibility candidates $\cC$ and all variables $x$, $x \in \cC$.
\end{enumerate}
\end{lemma}

\subsubsection*{Proof of Lemma~\ref{lem:redCR}}
\[
\mbox{For each type $\tau$,  $\Red{\tau} \in \CR$.}
\]
We prove that $\Red{\tau}$ satisfies the CR conditions by structural induction on $\tau$. We only consider the $\Red{\Box \tau}$ case and invite the interested reader to see~\cite{proofstypes} for atomic and function types.

\begin{description}
\item[CR1] given $s \in \Red{\Box \tau}$, we prove that $\tlet(u := s, v) \in \SN$, with $u \neq v$ (this follows from the definition of $\Red{\Box \tau}$, taking $\cC = \SN$, $t = v$, and noticing that $v \subst{u}{s'} = v$ for all $s' \in \Red{\tau}$). Then $s$ is a subterm of a strongly normalizing term, thus $s \in \SN$ as well.

\item[CR2] given $s \in \Red{\Box \tau}$ and $r$ s.t. $s \red r$, we need to prove that $r \in \Red{\Box \tau}$. By the definition of reducibility, we need to show that for all $\cC$ and for all $t \in \subCR{\cC}{u}{\Red{\tau}}$, $\tlet(u := r, t) \in \cC$. Since $s \in \Red{\Box \tau}$, by definition we have $\tlet(u := s, t) \in \cC$ and, since $\cC \in \CR$, we can perform a one step reduction and obtain $\tlet(u := r, t) \in \cC$ as needed.

\item[CR3] given $s$ neutral such that whenever $s \red s'$ we have $s' \in \Red{\Box \tau}$, we need to prove that $s \in \Red{\Box \tau}$: in other words, we assume $t \in \subCR{\cC}{u}{\Red{\tau}}$ for some candidate $\cC$, and we need to prove that $\tlet(u := s, t) \in \cC$. Since $\subCR{\cC}{u}{\Red{\tau}}$ is a pre-candidate (Lemma~\ref{lem:subCRisPCR}), by CR1 $t \in \SN$ and we can reason by well-founded induction. We consider all possible reducts of $\tlet(u := s, t)$:
  \begin{itemize}
  \item $\tlet(u := s, t) \red \tlet(u := s', t)$ with $s \red s'$: we know 
  by hypothesis that $s' \in \Red{\Box \tau}$, thus we have $\tlet(u := s', t) 
  \in \cC$ by definition of reducibility;
  \item $\tlet(u := s, t) \red \tlet(u := s, t')$ with $t \red t'$: by CR2 on 
  pre-candidates we know that $t' \in 
  \subCR{\cC}{u}{\Red{\tau}}$ and we obtain by the inner induction hypothesis on 
  $t'$ that $\tlet(u := s, t') \in \cC$;
  \item no other possible reduct, since $s$ is neutral.
  \end{itemize}
Then by CR3 on $\cC$, we know $\tlet(u := s, t) \in \cC$, as needed.
\end{description}

\begin{lemma}\label{lem:redbang}
$s \in \Red{\tau} \Longrightarrow \bang s \in \Red{\Box \tau}$
\end{lemma}
\begin{proof}
To prove that $\bang s \in \Red{\Box \tau}$, we need to show that for all choices of $\cC$ and $t \in \subCR{\cC}{u}{\Red{\tau}}$, we have $\tlet(u := \bang s, t) \in \cC$. Since $s \in \Red{\tau}$ and $t \in \cC$, by CR1 both terms are strongly normalizing, and we can reason by well-founded induction: then we consider all reducts of $\tlet(u := \bang s, t)$:
\begin{itemize}
\item $\tlet(u := \bang s', t)$ with $s \red s'$: we know $s' \in \Red{\tau}$ (by CR2), then by induction hypothesis $\tlet(u := \bang s', t) \in \cC$;
\item $\tlet(u := \bang s, t')$ with $t \red t'$: by CR2 on pre-candidates we prove $t' \in \subCR{\cC}{u}{\Red{\tau}}$, then by induction hypothesis $\tlet(u := \bang s, t') \in \cC$;
\item $t\subst{u}{s}$: since $s \in \Red{\tau}$ and $t \in \subCR{\cC}{u}{\Red{\tau}}$, we have $t\subst{u}{s} \in \cC$.
\end{itemize}
Then we obtain $\tlet(u := \bang s, t) \in \cC$ by CR3.
\end{proof}

\begin{lemma}\label{lem:redti}
$(\forall \psi \in \dom(\theta), \theta(\psi) \in \Red{\cT^\sigma(\psi)}) \Longrightarrow \tinsp(\theta) \in \Red{\sigma}$
\end{lemma}
\begin{proof}
Similarly to the previous lemma, all the terms in the codomain of $\theta$ (which form a finite tuple) are strongly normalizing, thus we can reason by well-founded induction. We thus consider all reducts of $\tinsp(\theta)$:
\begin{itemize}
\item $\tinsp(\theta')$ where $\theta'$ is obtained from $\theta$ by means of a one-step reduction of a single term in its codomain, keeping all the other terms the same: we easily prove that for all $\psi$, $\theta'(\psi) \in \Red{\cT^\sigma(\psi)}$ (using CR2), then by induction hypothesis we have $\tinsp(\theta') \in \Red{\sigma}$;
\item $q \theta$: by an inner induction on the structure of $q$, we prove that $q \theta \in \Red{\sigma}$.
\end{itemize}
We thus obtain $\tinsp(\theta) \in \Red{\sigma}$ by CR3.
\end{proof}

We now prove strong normalization for \JLs. The theorem can be stated in a more elegant way by means of the following definition.

\begin{definition}[adapted interpretation]
Given contexts $\Delta$ and $\Gamma$, a sequence of substitutions $\vect{\eta} = \eta_1,\ldots,\eta_n$ is a \emph{$(\Delta;\Gamma)$-adapted interpretation} (notation: $\vect{\eta} \vDash \Delta;\Gamma$) if and only if:
\begin{itemize}
\item $\dom(\vect{\eta}) = \dom(\Delta) \cup \dom(\Gamma)$
\item for all $u \in \dom(\Delta)$ we have $\vect{\eta}(u) \in \Red{\Delta(u)}$
\item for all $a \in \dom(\Gamma)$ we have $\vect{\eta}(a) \in \Red{\Gamma(a)}$.
\end{itemize}
\end{definition}

\subsubsection*{Proof of Theorem~\ref{thm:snmain}}
\[
\mbox{If $\Delta; \Gamma \vdash s : \tau$, then for all $\vect{\eta} \vDash \Delta;\Gamma$ we have $s \vect{\eta} \in \Red{\tau}$.}
\]
By induction on the typing derivation. We consider the interesting cases:
\begin{itemize}
\item $s = \bang s'$, $\tau = \Box \sigma$ and, by induction hypothesis, for all ($\Delta;\cdot$)-adapted interpretations $\vect{\eta'}$ ($\bang$ is opaque to truth variables, so the corresponding context is empty) $s' \vect{\eta'} \in \Red{\sigma}$. We choose $\vect{\eta'}$ to be equal to $\vect{\eta}$ restricted to $\Delta$: then $(\bang s) \eta = \bang (s \eta')$ and the thesis follows from the induction hypothesis and Lemma~\ref{lem:redbang}.
\item $s = \tlet(u := R, S)$, and (by induction hypotheses), $R \eta' \in \Red{\Box \sigma}$ (for all ($\Delta;\Gamma$)-adapted $\eta'$) and $S \eta'' \in \Red{\tau}$ (for all ($\Delta,u::\sigma;\Gamma$)-adapted $\eta''$), with $u$ chosen to be fresh. Since clearly $R\vect{\eta} \in \Red{\Box \sigma}$, by the definition of $\Red{\Box \sigma}$ it suffices to show that $S\vect{\eta} \in \subCR{(\Red{\tau})}{u}{\Red{\sigma}}$ to prove $\tlet(u := R, S)\eta = \tlet(u := R\eta, S\eta) \in \Red{\tau}$. This is equivalent to proving that $S\vect{\eta} \in \Red{\tau}$ and $S\vect{\eta}\subst{u}{t} \in \Red{\tau}$ for all $t \in \Red{\sigma}$: both statements are implied by the induction hypothesis on $S$.
\item $s = \tinsp(\theta)$ and (by induction hypothesis) for all $\psi$, for all ($\Delta,\cdot$)-adapted interpretations $\eta'$ we have $\theta(\psi) \in \Red{\cT^\tau{\psi}}$. We choose $\eta'$ to be the suitable restriction of $\eta$ and we obtain $\tinsp(\theta)\eta = \tinsp(\theta\eta') \in \Red{\tau}$ by Lemma~\ref{lem:redti}. 
\end{itemize}

The two corollaries we are interested in follow immediately.

\begin{corollary}[strong normalization for \JLs]\label{cor:snJLs}
If $\Delta;\Gamma \vdash_\JLs s : \tau$, then $s \in \SN$.
\end{corollary}
\begin{proof}
Immediately from Theorem~\ref{thm:snmain}, by replacing $\vect{\eta}$ with the identity substitution.
\end{proof}


\begin{corollary}[strong normalization for \JLr]\label{co:snJLr}
If $\Delta;\Gamma \vdash_\JLr M : A | s$, then $M \in \SN$.
\end{corollary}
\begin{proof}
By combining Corollary~\ref{cor:snJLs} with Lemma~\ref{lem:erasetchk} and Lemma~\ref{lem:erasered}.
\end{proof}

\end{document}


%% file: macro.tex
\usepackage{url}
\usepackage{amsmath,amssymb,amsthm}
\usepackage{xspace}
\usepackage{stmaryrd}
\usepackage{hyperref}
\usepackage{bussproofs}
\usepackage{graphicx}
\usepackage{bm}

\hypersetup{
    colorlinks=true,
    citecolor=blue,
    linkcolor=blue,
    urlcolor=blue
}

\newcommand{\subst}[2]{\left[{#2}\middle/{#1}\right]}
\newcommand{\lsubst}[2]{\ltimes\subst{#1}{#2}}
\newcommand{\csubst}[2]{\times\subst{#1}{#2}}

\newcommand{\pair}[1]{\langle {#1} \rangle}
\newcommand{\vect}[1]{\overrightarrow{#1}}

\newcommand{\bbN}[0]{\mathbb{N}}

\newcommand{\cC}[0]{\mathcal{C}}
\newcommand{\cD}[0]{\mathcal{D}}
\newcommand{\cE}[0]{\mathcal{E}}
\newcommand{\cF}[0]{\mathcal{F}}

\newcommand{\cN}[0]{\mathcal{N}}
\newcommand{\cP}[0]{\mathcal{P}}
\newcommand{\cQ}[0]{\mathcal{Q}}
\newcommand{\cR}[0]{\mathcal{R}}
\newcommand{\cS}[0]{\mathcal{S}}
\newcommand{\cT}[0]{\mathcal{T}}

\newcommand{\CR}[0]{\cC\cR}
\newcommand{\PCR}[0]{\cP\cC\cR}
\newcommand{\SN}[0]{\cS\cN}
\newcommand{\NT}[0]{\cN\cT}
\newcommand{\subCR}[3]{{{#1}_{#2}^{#3}}}
\newcommand{\Red}[1]{\mathsf{Red}_{#1}}
\newcommand{\derives}[0]{\mathrel{\triangleright}}
\newcommand{\tinsp}[0]{\iota}
\newcommand{\red}[0]{\mathrel{\leadsto}}

\newcommand{\qrefl}[0]{\mathbf{r}}
\newcommand{\qsym}[0]{\mathbf{s}}
\newcommand{\qtrans}[0]{\mathbf{t}}
\newcommand{\qapp}[0]{\mathbf{app}}
\newcommand{\qlam}[0]{\mathbf{lam}}
\newcommand{\qlet}[0]{\mathbf{let}}
\newcommand{\qti}[0]{\mathbf{ti}}
\newcommand{\qtr}[0]{\mathbf{trpl}}
\newcommand{\qoti}[0]{\mathbf{ti}}
\newcommand{\qotr}[0]{\mathbf{trpl}}
\newcommand{\qiti}[0]{\mathbf{iti}}
\newcommand{\qitr}[0]{\mathbf{itb}}

\newcommand{\qba}[0]{{\bm{\beta}}}
\newcommand{\qbb}[0]{{\bm{\beta_\Box}}}

\newcommand{\qd}[0]{\mathbf{d}}

\newcommand{\AU}[1]{\left\llbracket{#1}\right\rrbracket}

\newcommand{\code}[0]{{\mathrm{cd}}}
\newcommand{\src}[0]{{\mathrm{src}}}
\newcommand{\tgt}[0]{{\mathrm{tgt}}}

\newcommand{\CTX}{\Delta;\Gamma}
\def\orelse{\mathbin{|}}
\newcommand{\tchk}[4]{{#1} \vdash {#2} : {#3} \orelse {#4}}
\newcommand{\tchkh}[4]{{#1} \vdash {#2} : {#3} \orelse {#4}}
\newcommand{\qchk}[3]{{#1} \vdash {#2} : {#3}}
\newcommand{\qchkh}[3]{{#1} \vdash {#2} : {#3}}

\newcommand{\Eq}{\mathsf{Eq}}

\newcommand{\bang}{\mathord{!}}
\newcommand{\imp}{\supset}
\newcommand{\direq}{\mathrel{\overset{\triangleright}{=}}}
\newcommand{\equest}{\mathop{\overset{?}{=}}}
\newcommand{\Intro}{{I}}
\newcommand{\Elim}{{E}}

\DeclareMathOperator{\tlet}{let}

\DeclareMathOperator{\dom}{dom}

\def\iftrm{\mathsf{if}}

\newcommand{\JLh}{{\ensuremath{\lambda^h}}\xspace}
\newcommand{\JLr}{{\ensuremath{\lambda^{hc}}}\xspace}
\newcommand{\JLs}{{\ensuremath{\lambda^{hs}}}\xspace}

\newcommand{\impp}{\supset}

\def\rVS{\noalign{\vskip 2ex plus.4ex minus.4ex}}

\allowdisplaybreaks[1]